\documentclass[letterpaper,twocolumn,10pt]{article}
\usepackage{usenix-2020-09}

\newif \ifcomments \commentstrue
\newif \iffull \fullfalse

\usepackage{xcolor}
\usepackage{url}

\definecolor{newdiffcolor}{rgb}{0, 0, 0}
\newcommand{\newdiff}[1]{{\color{newdiffcolor} #1}}

\usepackage{amsfonts,amsmath}
\usepackage{amsthm}

\ifcomments
	\newcommand{\mahimna}[1]{\textsf{\small{\color{violet!80}{[Mahimna: {#1}]}}}}
\else
	\newcommand{\mahimna}[1]{}
\fi

\iffull \newcommand{\mypara}{\paragraph}
\else \newcommand{\mypara}[1]{\smallskip\noindent\textbf{#1}\;} \fi

\renewcommand{\paragraph}[1]{\smallskip\noindent\textbf{#1}\;}

\iffull \newtheorem{theorem}{Theorem}[section]
\else \newtheorem{theorem}{Theorem} \fi
\newtheorem{corollary}{Corollary}[theorem]

\theoremstyle{definition}
\newtheorem{definition}[theorem]{Definition}

\usepackage{graphicx} 
\usepackage[utf8]{inputenc}
\usepackage{changepage}
\usepackage{xurl}
\usepackage{hyperref}
\usepackage{xcolor}
\usepackage{amsfonts}
\usepackage{mathtools}
\usepackage{amsthm}
\usepackage{amssymb}
\usepackage{cleveref}
\usepackage{xspace}
\usepackage{array, caption, tabularx,  ragged2e,  booktabs}
\usepackage{physics}
\usepackage{amsmath}
\usepackage{tikz}
\usepackage{mathdots}
\usepackage{yhmath}
\usepackage{cancel}
\usepackage{color}
\usepackage{siunitx}
\usepackage{array}
\usepackage{multirow}
\usepackage{gensymb}
\usepackage{tabularx}
\usepackage{extarrows}
\usepackage{booktabs}
\usetikzlibrary{fadings}
\usetikzlibrary{patterns}
\usetikzlibrary{shadows.blur}
\usetikzlibrary{shapes}
\usetikzlibrary{positioning,arrows}

\usepackage{algorithm}
\usepackage{algpseudocode}
\usepackage{tablefootnote}
\usepackage{enumitem}
\usepackage[many]{tcolorbox}
\usepackage{makecell}

\usepackage[available,functional,reproduced]{usenixbadges}

\tcbset{
    sharp corners,
    colback = white,
    before skip = 0.2cm,    
    after skip = 0.5cm      
}                           %

\ifcomments
    \newcommand{\ari}[1]{{\small\textsf{\color{red}{[Ari: {#1}]}}}}
    \newcommand{\andres}[1]{{\small\textsf{\color{blue}{[Andres: {#1}]}}}}
    \newcommand{\james}[1]{{\small\textsf{\color{olive}{[James: {#1}]}}}}
    \newcommand{\kushal}[1]{{\small\textsf{\color{orange}{[Kushal: {#1}]}}}}
    \newcommand{\sarah}[1]{{\small\textsf{\color{brown}{[Sarah: {#1}]}}}}
    \newcommand{\jay}[1]{{\small\textsf{\color{cyan}{[Jay: {#1}]}}}}
   
\else
    \newcommand{\ari}[1]{}
    \newcommand{\andres}[1]{}
    \newcommand{\mahimna}[1]{}
    \newcommand{\james}[1]{}
    \newcommand{\kushal}[1]{}
    \newcommand{\jay}[1]{}
    \newcommand{\sarah}[1]{}

\fi

\newcommand{\util}{{\textsf{util}}}
\newcommand{\true}{{\texttt{true}}}
\newcommand{\false}{{\texttt{false}}}
\newcommand{\tokens}
{{\textsf{tokens}}}
\newcommand{\bribe}
{{\textsf{bribe}}}
\newcommand{\vote}
{{\textsf{vote}}}

\newcommand{\qv}
{{\textsf{quad}}}
\newcommand{\players}{{\mathcal{P}}}
\newcommand{\player}{\ensuremath{P}}

\newcommand{\apath}{\mathcal{A}}

\newcommand{\indexshort}{\textsf{VBE}\xspace}
\newcommand{\oindexshort}{\textsf{oVBE}\xspace}
\newcommand{\indexlong}{Voting-Bloc Entropy\xspace}
\newcommand{\indexlongbold}{\textbf{V}oting-\textbf{B}loc \textbf{E}ntropy\xspace}
\newcommand{\hypothesisname}{Learning Hypothesis of DAOs\xspace}
\newcommand{\hypothesisnameshort}{LHD\xspace}
\newcommand{\clustershort}{$\epsilon$-TOC\xspace}

\newcommand{\sgn}[1]{\text{sgn}(#1)}

\newcommand{\votehistory}{\mathcal{H}}

\theoremstyle{definition}
\newtheorem{exmp}{Example}

\newtheorem{thm}{Theorem}[section]

\newtcolorbox{boxA}{
    fontupper = \bf,
    boxrule = 1.5pt,
    colframe = black %
}

\newlength{\saveparindent}
\newlength{\saveparskip}
\newcounter{ctr}

\newenvironment{newitemize}{%
\begin{list}{\mbox{}\hspace{5pt}$\bullet$\hfill}{\labelwidth=15pt%
\labelsep=5pt \leftmargin=20pt \topsep=3pt%
\setlength{\listparindent}{\saveparindent}%
\setlength{\parsep}{\saveparskip}%
\setlength{\itemsep}{0.8pt} }}{\end{list}}

\begin{document}

\date{}

\title{Voting-Bloc Entropy: A  New Metric for DAO Decentralization}

\author{
  \rm Andrés Fábrega$^{1,2}$, Amy Zhao$^{2}$, Jay Yu$^{4}$, James Austgen$^{1,2}$, Sarah Allen$^{2, 3}$, \\ \rm Kushal Babel$^{1, 2}$, Mahimna Kelkar$^{1,2}$, Ari Juels$^{1, 2}$\\
  $^{1}$  Cornell Tech \hspace*{2.5em} $^{2}$ IC3 \hspace*{2.5em} $^{3}$ Flashbots \hspace*{2.5em} $^{4}$ Stanford University
} 

\maketitle

\begin{abstract}
  Decentralized Autonomous Organizations (DAOs) use smart contracts to foster communities working toward common goals. Existing definitions of decentralization, however---the `D' in DAO---fall short of capturing the key properties characteristic of diverse and equitable participation. 

This work proposes a new framework for measuring DAO decentralization called \indexlongbold (\indexshort, pronounced ``vibe''). \indexshort is based on the idea that voters with closely aligned interests act as a centralizing force and should be modeled as such. \indexshort formalizes this notion by measuring the similarity of participants’ utility functions across a set of voting rounds. Unlike prior, ad hoc definitions of decentralization, \indexshort derives from first principles: We introduce a simple (yet powerful) reinforcement learning-based conceptual model for voting, that in turn implies \indexshort.

We first show \indexshort's utility as a theoretical tool. We prove a number of results about the (de)centralizing effects of vote delegation, proposal bundling, bribery, etc.~that are overlooked in previous notions of DAO decentralization. Our results lead to practical suggestions for enhancing DAO decentralization. 

We also show how \indexshort can be used empirically by presenting measurement studies and \indexshort-based governance experiments. We make the tools we developed for these results available to the community in the form of open-source artifacts in order to facilitate future study of DAO decentralization.

\end{abstract}

\section{Introduction}
\label{sec:introduction}
A Decentralized Autonomous Organization (DAO) is an entity or community that operates based on rules encoded and executed on a public blockchain~\cite{buterin2013bootstrapping,hassan2021decentralized}. DAOs can serve many goals, including 
investment~\cite{Jentzsch2016decentralized}, 
grant distribution~\cite{MolochDAO:2023}, gaming-guild 
organization~\cite{Avocado:2023,GuildFi:2023}, and ecosystem
governance~\cite{Uniswap:2023,Optimism:2023}. DAOs play a prominent role in blockchain ecosystems, and are rising rapidly in popularity: at the time of writing (Aug.~2024), the aggregate value across all DAO treasuries exceeds \$22 billion~\cite{DeepDAO:2023}.

As the name suggests, a DAO's governance is decentralized, meaning that decision-making does not rely on a single individual or highly concentrated authority---in contrast to, e.g., a corporation, where a CEO and board of directors make major decisions. Instead, decisions in a DAO are typically made through community votes on proposals, the outcomes of which are enforced automatically by the blockchain in which the DAO's rules, i.e., its \emph{smart contract}, reside. 

Decentralization is a core property that DAOs strive for, as diverse and equitable participation are fundamental ideals in these communities (and the blockchain ecosystem more broadly)~\cite{daos}. Most DAOs have their own associated crypto assets (or ``tokens'') and weigh voting power by token holdings. However, it is common for vote outcomes to be determined by a small set of ``whales''---a colloquial term used to denote the largest token holders. Such centralization, as well as low voting participation, are a pervasive source of concern in DAO communities. 

Decentralization is of critical importance and concern not just when it comes to successful governance of DAOs but also DAO \emph{security}. As noted in~\cite{feichtinger2024attacks}, ``if a large (delegated) token supply is held only by a few addresses
or entities, many attack vectors become more likely to succeed.'' Poorly decentralized DAOs have historically been at higher risk of various governance attacks, including cabals pushing through adversarial proposals, e.g.,~\cite{GoldenBoys:2024}, attackers plundering treasuries~\cite{feichtinger2024attacks,Malwa:2023}, and systemic voter bribery~\cite{bribery}---for which active marketplaces exist today~\cite{lloyd2023emergent}.

Academic works and DAO community participants have studied DAOs~\cite{fritsch2022analyzing,feichtinger2023hidden,feichtinger2024attacks,philquadvotingdaos,sharma2024unpacking} and recommended ways to improve their decentralization~\cite{Machiavelli:2023}. To do so effectively, though, requires an ability to \textit{model} and \textit{measure} decentralization in a  way that is reflective of a broad set of real-world concerns. These requirements motivate our work in this paper.

\paragraph{DAO decentralization: previous attempts.} The most common basis for evaluating decentralization in DAOs and other blockchain settings is \textit{token ownership}, specifically the distribution of assets and consequently voting rights among participants~\cite{sharma2024unpacking,fritsch2022analyzing}. Informally, concentration of a large fraction of tokens in a small number of hands---and thus the ability of a small group to determine voting outcomes---is indicative of strong centralization. More widespread distribution, conversely, suggests decentralization. Prior decentralization measures generally formalize this intuition by computing a particular function over the distribution of tokens among individual members of the DAO, such as entropy,\footnote{Entropy is typically defined over a random variable. A token ownership distribution may be viewed as a random variable for an experiment where a token is selected uniformly at random and its owner is output.} the Gini coefficient~\cite{Gini:2023} or the Nakamoto 
coefficient~\cite{nakamotocoefficient}.

Token ownership distribution across individual accounts has serious shortcomings as a framework for decentralization, however. To begin with, it is visible on chain only in terms of per-address holdings, not control by real-world individuals. For instance, an individual who holds 51\% of tokens in a DAO, but spreads them among a large number of addresses, could create an appearance of decentralization while having majority control. Even if tokens are held by distinct entities, a notion put forward in, e.g.,~\cite{karakostas2022sok}, those entities may have \emph{aligned interests and act in concert}---a form of centralization. The following examples illustrate cases in which a DAO may be strongly centralized, \textit{even if token ownership appears to imply strong decentralization}.

\begin{exmp}[Low participation / apathy]
\label{exmp:participation}
Lack of participation in DAO governance votes is widespread in practice~\cite{daoapathy} and induces a form of centralization. Consider, for example, a DAO proposal where half of voters do not vote and voters other than whales vote ``yes'' by a 2:1 margin. Whales with just 12.6\% of all tokens can swing the vote and force a ``no'' vote---an example of centralized power. 
\end{exmp}
    
\begin{exmp}[Herding]
\label{exmp:herding}
Interviews with DAO participants reveal a tendency to vote in alignment with influential community members to preserve reputation~\cite{sharma2024unpacking}, as individual votes are today usually publicly observable. This effect---often called 
\textit{herding}~\cite{alon2012sequential}---has a centralizing effect. It aligns votes around the choices of a few participants. (This problem is similar to ``herding'' in 
classical voting theory~\cite{gonzalez2006herding}.)
\end{exmp}

\begin{exmp}[Bribery / vote-buying]
\label{exmp:bribery}
Bribery---specifically, \textit{vote-buying}---has been a longstanding concern of DAO organizers~\cite{bribery,Eluke:2024}. It has a centralizing effect, as it aligns voters around a choice dictated by the briber.
\end{exmp}

Recognizing that token-ownership alone doesn't give a full picture of decentralization, 
researchers have explored broader notions. Most notably, Sharma et al.~\cite{sharma2024unpacking} have considered entropy measures limited to those voters who participate in votes, and have also explored graph-based representations of voting patterns (degree centralization, degree assortativity, etc.).  
Token-ownership distribution among voting participants overlooks important issues, such as those in Examples~\ref{exmp:herding} and~\ref{exmp:bribery}, however, and it is unclear how to interpret graph-based empirical models. With no consensus in the community about how to measure DAO decentralization today, there is a lack of principled guidance on ways to improve DAO decentralization and to combat threats to decentralization, such as vote-buying.

\paragraph{\indexlong (\indexshort).} The primary contribution of our work is to introduce \textit{\indexlongbold}~(\indexshort, pronounced ``vibe''), a decentralization measure tailored to DAO governance. \indexshort is based on a core principle: individual voters with closely aligned interests across elections are a centralizing force---who operate in concert as single, abstract voting entities---and should be modelled as such. 
Expressed differently, the key idea in \indexshort is to \textit{define centralization as the existence of large voting blocs}.

Formally, we express this principle in terms of the \textit{utility functions}~\cite{Utility:2024} of DAO participants, i.e., quantification of the gain or loss associated with voting outcomes. For a given set of elections, a voting bloc is a cluster of voters whose utility functions are similar over outcomes. \indexshort then, measures entropy over voting blocs based on utility functions---rather than over individual token holdings. The result is a broad concept that captures the centralization embodied in all of our examples above. Thus, \indexshort is a more accurate reflection of decentralization than prior metrics, as it accounts for subtle, centralizing forces which are not apparent from token ownership alone, but which affect a DAO's \emph{true} degree of decentralization. These centralizing effects are reflected in the \emph{alignment of  voters' utility functions}, which \indexshort's clustering step is sensitive to. Note that \indexshort is in fact a framework: It allows different notions of clustering and entropy to be plugged in, and thus can be tailored to particular applications.

Unlike prior decentralization metrics, we arrive at \indexshort from first principles. Toward this end, we introduce the \emph{\hypothesisname (\hypothesisnameshort)}, a reinforcement learning (RL) conceptual model of voting in DAOs. The key insight in the \hypothesisnameshort is that the purpose of voting is to collectively \emph{explore a policy space}, with the goal of maximizing some global reward function, (e.g., the DAO's treasury), as well as individuals' local reward functions (e.g., their monetary returns).

We show how this process can be naturally modeled as an RL problem, specifically a \textit{multi-agent} RL or \textit{MARL} problem. This modeling---the basis for the \hypothesisnameshort---underpins our reasoning about DAO decentralization. Specifically, existing work on MARL stresses the importance of \emph{agent diversity}, an aspect of MARL that motivates our formulation of \indexshort.

\paragraph{\indexshort in theory.} 
\indexshort is an abstract metric: It cannot be measured \textit{directly}, since users do not typically express (or even know) their utility functions explicitly. That is, utility functions are so-called \textit{latent variables}~\cite{Latent:2023}, %
and thus \indexshort is as well. However, \indexshort provides an important basis for \textit{reasoning about the impact of policy choices on decentralization}. 

We use \indexshort to prove a number of simple theorems about how various practices might increase or decrease DAO decentralization. 
In some cases, our theorems capture intuitive or folklore notions of decentralization expressed by the community (e.g., the three examples mentioned above). In other cases, they offer \emph{new} insights about decentralization. For example, we show that \emph{as the decentralization of a DAO rises, so does the risk of systemic bribery---and vice versa}. This result--- alongside our theorem showing that the act of bribery decreases decentralization---sheds new light on the connection between bribery in DAOs and decentralization.

Most importantly, several of our theorems lead to actionable recommendations regarding DAO governance. Further, our theorems serve as examples of a higher-level contribution: \indexshort's utility as a \emph{theoretical tool} to derive formal results about decentralization. 

\paragraph{\indexshort in practice.} While latent variables are not directly measurable, they can be estimated via measurable quantities called \emph{observable variables}~\cite{Latent:2023}. As such, we introduce \emph{observable \indexshort (\oindexshort)}, an observable variable that can be used to estimate \indexshort. \oindexshort estimates voters' utility functions in terms of observable, on-chain data---such as voting history---and clusters voters using this data. Like \indexshort, \oindexshort is a framework that can be instantiated with various clustering and entropy notions. By preserving \indexshort's structure, \oindexshort inherits the benefits of our framework, particularly with an accurate estimate of utility functions (we discuss this extensively in Section~\ref{sec:vbe-in-practice}). 

Since \oindexshort \emph{is} directly measurable, it opens the door for \indexshort-based measurement studies and decentralization experiments. We report on two example use cases of \oindexshort.

First, we perform a measurement study of the historical \oindexshort across a number of popular DAOs, which we make available in the form of a public dashboard. Our dashboard contains a variety of per-DAO \oindexshort metrics, and updates weekly to reflect the evolving \oindexshort landscape. 

Second, we show how \oindexshort can serve as a \emph{metric in governance experiments} to understand the effect of a particular mechanism on a DAO’s decentralization, by presenting an example experiment taking place in an ongoing collaboration with the Optimism Collective, who are using \oindexshort as a decentralization signal in an upcoming governance experiment in RetroPGF~\cite{OptimismPGF}.

To perform these studies, we developed a suite of tools to process governance data and compute \oindexshort in a variety of settings. As an additional contribution, we make these artifacts available to the community in the form of a comprehensive, open-source \emph{\oindexshort toolkit} for further decentralization studies. Our toolkit---which can be run as a standalone program or integrated as a library---supports various instantiations of \oindexshort, and is structured in a modular way so that users can easily tweak the various parameters to fit their use case.  

\paragraph{Contributions.} In brief, our contributions in this work are:
\begin{newitemize}
    \item \textit{\hypothesisname} (\hypothesisnameshort): We introduce \hypothesisnameshort, a simple multi-agent RL-based model of DAO voting (Section~\ref{sec:rl-decentralization}). This model is of independent interest, as it enables principled study of other questions about DAO governance.
    \item  \textit{\indexlong (\indexshort)}: We use the \hypothesisnameshort to derive \indexshort, a new framework for DAO decentralization that generalizes prior metrics and addresses a number of their shortcomings (Section~\ref{sec:vbe-definition}). %
    \item \textit{Proving results about DAO governance}: We leverage \indexshort to prove results about the impact of DAO practices and designs on decentralization (Section~\ref{sec:vbe-in-theory}), in some cases reaffirming some folklore notions and in others revealing new insights about DAO decentralization. More broadly, these theorems are examples of \indexshort's utility as a theoretical tool.
    \item \textit{Empirical studies:} We introduce \oindexshort, a directly measurable variant of \indexshort. We perform a \oindexshort-based measurement study of popular DAOs, presenting a public dashboard with our results (Section~\ref{subsec:ovbe-landscape}). We also release an open-source toolkit for \oindexshort to stimulate future empirical work in this space, and report on real-world governance experiments using the toolkit (Section~\ref{subsec:ovbe-gov-experiments}). 
    \item \textit{Practical guidance}: Based on our theoretical and experimental results, we present and summarize concrete points of practical guidance for DAO design and deployment (Section~\ref{sec:guidance}). 
\end{newitemize}

The main conceptual contribution of our work is a new \emph{conceptual model} for DAO decentralization that captures voting entities---instead of individual accounts---characterized by their \textit{alignment of incentives}. From this foundational idea, we derive our theoretical and experimental results. While these results are independently interesting (and indeed yield the actionable lessons we summarize in Section~\ref{sec:guidance}), they showcase a broader point, namely, \indexshort's flexibility and utility as a powerful tool to understand decentralization, which we hope motivates future work in assessing both the effectiveness of DAO governance. \newdiff{Furthermore, since decentralization is a critical security feature for DAOs, an accurate model and measure for decentralization are of central importance to DAO security---both in terms of monitoring propensity to risk, as well as the development of new techniques aimed at enhancing security by increasing decentralization.}

\section{Voting as a Learning Problem}
\label{sec:rl-decentralization}
The goal of our work is to introduce a decentralization framework for DAOs that addresses the gaps of existing metrics. To do so, we take a different approach to this problem than that of prior works: rather than starting from a candidate decentralization metric, we first take a step back, and start from a more foundational question instead: What is a \emph{first principle} for voting in DAOs, from which we can then derive a decentralization metric? That is, what is DAO decentralization itself?

Towards this, we introduce the \hypothesisname (\hypothesisnameshort) in this section. The \hypothesisnameshort is a multi-agent reinforcement learning (MARL)-based model that characterizes voting in DAOs, thus serving as a foundational principle from which we can study DAO governance. We use the \hypothesisnameshort to cast DAO decentralization as an analogue to \emph{agent diversity} in MARL; this insight will serve as the starting point for \indexshort. \newdiff{This first-principles derivation thus supports VBE as a logical choice among possible metrics.} We discuss the main conceptual ideas behind our model in this section---which are sufficient to understand the rest of our work---and refer readers
to the extended version of our paper for more details (including additional background on MARL).

\paragraph{Brief background on MARL.} Reinforcement learning (RL)~\cite{RL:2024} is a field of machine learning that studies the behavior of an agent that interacts with an environment, with the goal of maximizing some long-term expected reward. Multi-agent reinforcement learning (MARL)~\cite{MARL:2024} is a particular type of RL, in which \emph{multiple} agents are present in a shared environment.

MARL is a broad framework that can be used to model a multitude of problems. A MARL model has a few key components. There is an \emph{environment}, and a set of \emph{states} for this environment. There is a set of \emph{agents} present in this environment, each with a set of \emph{actions}, which defines the ways they can interact with the environment. The environment transitions from one state to the next as a result of the agents' collective actions; this transition is modeled via a \emph{state transition function}, which denotes the probability of transitioning from one state to another given the agents' actions. Whenever the environment makes a state transition, each agent receives an immediate reward, which is modeled via a \emph{reward function}.

The model proceeds in rounds. In each iteration, the agents observe the environment's current state and each receive a reward, from which each chooses its next action. This vector of actions is fed back to the environment, which transitions to a new state (as per the state transition function), resulting in new rewards for the agents (as per the reward functions). This process continues iteratively.

The actions taken by each agent are defined by a \emph{policy}, which denotes the probability that the agent takes a particular action in a particular state. The quality of a ``policy'' can be expressed in terms of a \emph{state-value function}, which represents the long-term expected return of executing the policy for many time-steps. The goal of each agent, then, is to find an \emph{optimal policy} that maximizes the value function. Finding this optimal policy requires a balance between \emph{exploration} and \emph{exploitation}: the agent needs to search for new potential strategies while also taking advantage of its present understanding of the environment.

The relationship between the reward functions of the agents is a key property of MARL models. All agents can \emph{cooperate}, which means that they have the same reward functions. Two agents can \emph{compete} against each other, which means that their reward functions are the negative of each other. More generally, agents may have related but different reward functions, i.e., there are elements of both cooperation and competition (which is typically referred to as a ``general-sum game'').

\mypara{Diversity in MARL.} Diversity of viewpoints is well recognized as a key aspect of collective decision making, as the ``wisdom of the crowds'' often leads to better decisions than singular viewpoints~\cite{surowiecki2005wisdom}. Similarly, in the context of MARL, \emph{diversity} between individual agents (particularly in the fully or partially cooperative setting) is an important property of a model that leads to more efficient learning in many contexts.

Agent diversity is a loose term that captures the degree of heterogeneity between agents, which can be analyzed from many angles (e.g., roles~\cite{wang2020roma}, actions~\cite{jaques2019social}, rewards~\cite{jiang2021emergence}, policies~\cite{mckee2020social}, etc.). Diversity has many benefits, including more efficient exploration~\cite{li2021celebrating}, development of more advanced cooperation policies~\cite{li2021celebrating}, discovery of niche skills~\cite{perez2021modelling}, etc. As such, a number of diversity-boosting techniques have been put forth in the MARL literature, leading to experimental validation that diversity improves the collective performance of a multi-agent 
model~\cite{mckee2020social,li2021celebrating,lee2019learning}. Indeed, agent diversity has led to improved performance in many concrete applications of MARL, such as investment portfolio management~\cite{lee2020maps}, multi-robot systems~\cite{bettini2023heterogeneous}, 
autonomous driving~\cite{pmlr-v155-zhou21a}, etc.

To study the benefits of agent diversity for model performance, it is important to have a way to \emph{measure} diversity to begin with. As such, there are many proposed diversity measures in the MARL literature. While an in-depth systematization of diversity metrics is outside the scope of this work, below we provide a high-level overview of diversity metrics. 

\paragraph{A (short) taxonomy of MARL diversity metrics.} We performed a literature review of diversity metrics in MARL to understand their similarities and differences. We discuss our main takeaways here, and refer interested readers to, e.g.,~\cite{liu2022unified,liu2021towards} for more technical details on diversity metrics, which are outside the scope of our work.

Conceptually, existing MARL diversity metrics have two main ingredients. First, some property of the agents is selected as the basis for gauging diversity (e.g., actions, reward functions, value functions, policies, etc). This is followed by computing some function over the the distribution of this property across agents (e.g., a statistical distance measure, correlation coefficient, entropy measure, etc.). Examples of mathematical objects that are used as the basis for agent diversity include policy 
functions~\cite{lupu2021trajectory,hu2022policy,zhao2023maximum}, reward
functions~\cite{perez2021modelling,zhang2021coordination}, actions~\cite{mckee2022quantifying}, and state or action-value functions~\cite{hu2022policy, yang2020multi}. Then, examples of functions that are computed over these objects are entropy~\cite{zhao2023maximum}, Jensen-Shannon divergence~\cite{lupu2021trajectory}, total variation distance~\cite{mckee2022quantifying}, KL divergence~\cite{hu2022policy,zhao2023maximum}, and the Pearson correlation coefficient~\cite{zhang2021coordination}.

While diversity metrics vary a lot in their details---and there seems to be no consensus in the community for what is the best metric---our literature review revealed an important conceptual lesson: diversity metrics characterize differences in \emph{behavior} (e.g., actions or policies) and \emph{incentives} (e.g., reward or value functions) of the agents, instead of just individuality of the agents. Indeed, works in this space have emphasized the fact that ``differences'' is not equivalent to ``diversity'', and how only optimizing for the former can lead to, for example, learning circular behaviors~\cite{yang2021diverse}. This insight, in addition to the two-point structure of agent diversity metrics, will form the basis for \indexshort.

\paragraph{\hypothesisname (\hypothesisnameshort).} We now describe the \emph{\hypothesisname (\hypothesisnameshort)}, a simple MARL model for voting in DAOs. This model is based on the insight that the fundamental goal of voting is to collectively \emph{explore a policy space} in order to maximize a set of related objectives. While members of a DAO have individual objectives (e.g., maximizing their monetary holdings), these are related to each other by the fact that (by definition) DAO members have shared assets among themselves, e.g., the DAO's token. Thus, they share a global goal, such as maximizing the value of this token or the DAO's treasury. To pursue these goals, the members collectively perform iterative ``actions'' in some policy space, as determined by the proposals and the resulting votes. We can frame this process as a simple MARL problem, which is what we refer to as the \hypothesisname (\hypothesisnameshort). The \hypothesisnameshort serves as a foundational model for DAO voting.

The environment corresponds to the blockchain, and each state corresponds to a state of the blockchain at the time when an election is taking place. The agents correspond to the members of the DAO. At each iteration, each agent has three available actions corresponding to voting choices for the election in question. The collective action of the multi-agent model consists of the vector of votes of all players, and the state transition function then moves to the next blockchain state as per the outcome of the vote, based on the DAO's voting system. After each state transition, the reward for each agent corresponds to their monetary utility for that election.

The goal of each agent is to maximize their long-term expected rewards. These rewards, however, are underpinned by some common objective, which broadly speaking represents the well-being of the DAO. As such, \hypothesisnameshort represents a general-sum game, were players need some level of coordination in order to achieve individual benefit.

\paragraph{DAO decentralization as agent diversity.} Based on the model of DAO members as MARL agents in \hypothesisnameshort, we can \emph{frame decentralization in DAOs as equivalent to agent diversity in MARL}. Diversity reflects the \emph{true} heterogeneity of MARL agents, i.e., the fundamental differences between these that lead to meaningful improvements in exploration, contribute new perspectives, and lead to greater collective wisdom. These are the same set of principles that fundamentally represent decentralization in DAOs, thus providing a new lens through which to study decentralization. Based on this insight, we can leverage the literature on diversity metrics to inspire a construction for a DAO decentralization metric; while the specific diversity metrics in MARL are highly bespoke for RL and application-specific, we can base our construction on the core principles behind these. This leads us to \indexshort, which is the main contribution of our work.

\section{\indexlong (\indexshort)}
\label{sec:vbe-definition}
We introduce \textit{\indexlong}~(\indexshort) in this section, our new framework for decentralization that generalizes prior metrics and sidesteps their limitations. It does so by normalizing token holdings based on voters' utility functions.

\paragraph{\indexshort: core ideas.} The starting point for \indexshort is the main conceptual takeaway we gleaned from our study of MARL diversity metrics: an accurate measure of decentralization should be reflective of the different \emph{behaviors} and \emph{incentives} among agents (voters), instead of merely identifying that the agents are different entities. This directly leads to the key idea behind \indexshort: instead of modeling decentralization in terms of the distribution of tokens across individual voters (as prior metrics do), we frame it instead in terms of the distribution of tokens across \emph{groups of voters with aligned interests across elections}, which are functionally acting as a single entity. Similar to diversity in MARL, aggregating voters based on aligned interests allows us to capture interactions and relationships among players in the system, which determine the true degree of decentralization of a DAO. 

We formalize the notion of ``aligned interests'' in terms of the DAO members' \emph{utility functions}~\cite{Utility:2024} across elections. In our setting, this is the natural mathematical object that represents the voters' incentives across elections (and, indeed, corresponds to reward functions in \hypothesisnameshort, which is a common first-ingredient in MARL diversity metrics). \indexshort, then, computes some function over these utilities in order to partition the set of players, which ultimately serves as the basis upon which to gauge the distribution of tokens. Thus, conceptually, \indexshort generalizes existing decentralization metrics by bootstrapping the two-point structure of MARL diversity metrics as a \emph{pre-processing} step to determine the groups---or \emph{blocs}---of voters with aligned incentives, over which we can then compute the entropy (or other similar function) of tokens.

\paragraph{DAO abstraction.} Before presenting the formal definition of our \indexshort framework, we introduce the notation that our definition and theorems rely on.

Let $\players = \{\player_1, \ldots, \player_n\}$ be the set of token holders in a system, and $\tokens \colon \players \to \mathbb{R}^+$ a mapping specifying the number of tokens held by each $\player \in \players$. (We will often overload this notation, and input a \emph{set} of accounts to $\tokens$ instead, by which we mean the total tokens held across all accounts in the set). These token holders participate in a set of (binary) elections $E = \{e_1, e_2, \ldots, e_m\}$, where we denote by $\vote_\player \colon E \to \{\true, \false, \bot\}$ player $\player$'s vote in election $e$; $\bot$ indicates that $\player$ abstained from voting in $e$. We represent all of $\player$'s votes across $E$ by $V_{E, \player}$. We define $\util_\player \colon E \times \{\true, \false\} \to \mathbb{R}$ to be the monetary utility of an outcome of $\true$ or $\false$ in $e$ to player $\player$, where we make the simplifying assumption that $\util_\player(e, \true) = -\util_\player(e, \false)$. Player $\player$'s total utility across all elections $E$ is represented by a vector $U_{E, \player} \coloneqq (\util_\player(e_i, \true))_{i \in [m]} \in \mathbb{R}^m$; we denote by $U_{E, \players}$ all players' utilities, i.e., $U_{E, \players} \coloneqq (U_{E, \player})_{\player \in \players}$.

Token holders often have low stakes in the elections, resulting in lack of interest or abstaining from voting altogether. More formally, we say that player $P$ is \emph{$\epsilon$-apathetic} in election $e$ if and only if $|\util_\player(e, \true)| \leq \epsilon$. We denote this set of apathetic voters by $\apath$. If the system supports vote delegation (for example, as a means to combat voter apathy), players may delegate their tokens to others, who cast a single vote on behalf of all the tokens they now~hold.

\subsection{Framework for \indexshort}\label{subsec:ovbe-framework}
We present \indexshort in this section. \indexshort is an abstract \emph{framework} for decentralization, which is parameterized by: (1) a clustering function, and (2) an entropy measure, which are the two key ingredients that underpin our definition.

\paragraph{Clustering.} We let $C \colon U_{E, \players} \times U_{E, \players} \to \{0, 1\}$ be a clustering function\footnote{Formally, a clustering function takes as input a set and returns a partition of it. For clarity of presentation, we represent $C$ as the induced equivalence relation of the partition, as these are corresponding terms.} that outputs $1$ if the utilities of two players are ``aligned'' across all elections $E$, and 0 otherwise. Our definition of \indexshort is agnostic to a specific clustering function, and instead only assumes that $C$ specifies an equivalence relation $\sim_{C}$ on the set  $U_{E, \players}$. Note that $C$ additionally induces a partition on $\players$, whereby $\player_i$ and $\player_j$ are in an equivalence class if and only if $C(U_{E, \player_i}, U_{E, \player_j}) = 1$. That is, $C$ partitions $\players$ into classes of players with aligned utility functions across elections. We will often overload notation and directly refer to $C$ as a partition of $\players$. Following standard notation, we denote that two players are in the same class by $P_i \sim_{C} P_j$, the set of all classes by $\players/\sim_{C}$, and the class $\player$ belongs to by $[\player]$.

Note that our model (and, thus, our definition of clustering) assumes that players vote on binary elections. This assumption is just for clarity of presentation. We could also consider clustering functions that partition players based on their utility functions over elections with more outcomes. (In fact, we explicitly do this in one of our experiments in Section~\ref{sec:vbe-in-practice}.)

\paragraph{Entropy.} The clustering step serves as a way to normalize token holdings based on the players' aligned incentives across a set of elections. After this, we can then, as is standard, compute some measure of the distribution of tokens, but this time across \emph{the resulting equivalence classes}.

More formally, we let $F : \mathbb{P} \times \players^{\mathbb{R}^+} \to \mathbb{R}$ be a function from the distribution of tokens across \emph{sets} of accounts to real numbers. In this notation, $\mathbb{P}$ denotes the set of all partitions of $\players$. So, the function $F$ takes as input a partition of $\players$ and the function $\tokens \colon \players \to \mathbb{R}^+$ (which maps accounts to the tokens they own), and returns some real number. The purpose of $F$ is to measure, in some sense, how ``evenly distributed'' tokens are across voting blocs. For instance, $F$ can be any of the many variants of Rényi entropy\footnote{Entropy is formally defined over a random variable, but we are overloading notation to think of the mapping between sets of accounts and their respective cumulative token balances as the probability mass function of a random variable.} (e.g., min-entropy, Shannon entropy, or max entropy), or any of the token distribution functions found in prior decentralization metrics. (In particular, note that we can represent any prior decentralization metric as an instantiation of \indexshort, by using the vacuous clustering function wherein each bloc has one and only one account, over which we compute the distribution function used by the metric.) We stress, however, that in principle $F$ can be any function, and our definition makes no assumptions about its structure.

\medskip
We are now ready to define \indexshort. Intuitively, our definition says that a DAO is more decentralized if the distribution of tokens across the blocs specified by $\sim_{C}$ has high entropy according to $F$. More concretely:
\begin{definition}[\indexlong]
\label{def:govdec3}
For a set of elections $E$, a set of players $\players$ with corresponding utilities $U_{E, \players}$, a mapping specifying token ownership across accounts $\tokens$, a clustering metric $C$, and an entropy measure $F$, we define \textit{\indexlong~(\indexshort)} to be: 
\begin{equation*}
\indexshort_{C, F}(E,\players, U_{E, \players}, \tokens) \coloneqq F(\players/\sim_{C}, \tokens).\end{equation*}
\end{definition}

\medskip
Note that, since \indexshort is a framework, it is (more accurately) a \emph{family} of decentralization metrics, all under the umbrella of \indexshort's two-step structure. That is, concrete decentralization metrics are a result of \emph{instantiations} of \indexshort with particular clustering functions and entropy measures. In particular, \indexshort can be instantiated with, for example, state-of-the-art clustering functions (e.g., K-means, hierarchical clustering, DBscan, etc.), and commonly-used entropy measures (e.g., Shannon or min-entropy). The specific instantiation of \indexshort to use is problem-specific, analogous to how different clustering functions are better suited for different types of inputs and applications. \newdiff{(We discuss this in more detail in Section~\ref{sec:vbe-in-practice}).}

\paragraph{An overview of \indexshort uses.} As we emphasized in Section~\ref{sec:introduction}, (any instantiation of) \indexshort is a \emph{latent metric}, i.e., it cannot be measured directly. This is due to the fact that utility functions are latent variables~\cite{Latent:2023}---conceptually important, but not always directly measurable. Thus, a priori, \indexshort's primary use case is as a \emph{theoretical tool} to formally derive results about the impact of policy choices or practices in decentralization, for which conceptual reasoning about utility functions is sufficient. We discuss this use case in more detail in Section~\ref{sec:vbe-in-theory}.

The second application of \indexshort is as an \emph{empirical tool} for \indexshort-based measurement studies and governance experiments. While we cannot directly measure the ``true'' \indexshort of a DAO, we can \emph{estimate} it: utility functions (and, thus, \indexshort) can be approximated via measurable \emph{observable variables}~\cite{Latent:2023}. This variant of \indexshort, which we cover in depth in Section~\ref{sec:vbe-in-practice}, inherits all the conceptual benefits of the \indexshort framework while additionally being tractable.

\newdiff{
\subsection{\indexshort and DAO Security}\label{subsec:vbe-security}
While \indexshort is a decentralization metric, it has important connections to DAO \emph{security} as well, which we discuss in this section. Uses of \indexshort (which are the focus of subsequent sections) therefore have direct applications to DAO security.

\paragraph{Decentralization and DAO security.}
As discussed in Section~\ref{sec:introduction}, decentralization is a critical \emph{security} feature for DAOs: inadequate decentralization can lead to unfair extraction of value from DAOs by subgroups, and even outright theft in the form of \emph{governance attacks}. Notable examples of governance attacks include treasury raids~\cite{GoldenBoys:2024}, rug pulls~\cite{rugpulls}, systemic bribery~\cite{bribery}, flash loan attacks~\cite{dotan2023vulnerable}, and more. DAOs that are more centralized are at higher risk for such attacks~\cite{feichtinger2024attacks}. Adversaries often attack DAOs by acquiring sufficient voting power to perform a malicious action. Thus, more centralized DAOs---particularly those with high voter apathy---are especially vulnerable because of (1) a low threshold to trigger an attack, (2) easy formation of adversarial coalitions, and (3) difficulty in reverting the impact of attacks. A healthy degree of decentralization is thus critical for mitigating security risks.

\paragraph{VBE and DAO security enhancements.} The importance of decentralization for DAO security highlights the fact that an accurate decentralization \emph{measure} is a critical security tool. \indexshort can thus help bolster DAO security in two ways:

\begin{newitemize}
    \item \emph{Monitoring}: Continuous tracking of a DAO's decentralization helps monitor its security and propensity to risk. By offering a more accurate decentralization measure, \indexshort more accurately reflects a DAO's risk profile than previous measures. Importantly,  many governance attacks are underpinned by dangerous \emph{group alignment}, which \indexshort, unlike prior metrics, is sensitive to.
    \item \emph{Design}: \indexshort offers a means of assessing the impact of governance design choices aimed at increasing decentralization to enhance security. \indexshort can be used to formally prove the directional impact of policy choices on decentralization (Section~\ref{sec:vbe-in-theory}); and to experimentally test decentralization interventions aimed at increasing security (Section~\ref{sec:vbe-in-practice}).
\end{newitemize}

\medskip
Now that we have introduced the \indexshort framework, and its relevance to DAO security, the rest of our work presents and exemplifies its use cases.
}

\section{\indexshort in Theory}
\label{sec:vbe-in-theory}
We present a variety of theoretical results implied by \indexshort in this section. In some cases, our results are simple, and show how \indexshort confirms intuitive notions about decentralization expressed by the community. While simple, these results show how \indexshort is able to capture many of the subtle issues that impact decentralization in a DAO. In other cases, \indexshort provides novel insights on decentralization, from which we can derive practical guidance (see Section~\ref{sec:guidance}). More broadly, this section showcases \indexshort's utility as a theoretical instrument, and how it can serve as  formal groundwork for future results about decentralization. \newdiff{In particular, VBE can be used as a tool to assess the theoretical impact of governance design choices that are aimed at enhancing a DAO's security by increasing decentralization (Section~\ref{subsec:vbe-security}).}

Before presenting the implications of \indexshort, we first remark that for most ``reasonable'' instantiations of $F$ (such as Shannon or min-entropy), the ``trivial'' clustering function which assigns each player to its own cluster \emph{gives an upper bound on \indexshort}. Concretely, this fact holds for any $F$ that increases whenever the tokens held by any players's voting bloc increase: if this is the case, for all players in the system, the number of tokens held by their bloc according to any clustering metric is necessarily greater than or equal to the number of tokens held by a ``bloc'' that only contains themselves. In particular, recall from Section~\ref{sec:vbe-definition} that prior decentralization metrics can be cast in the \indexshort framework precisely as instantiations that use this trivial clustering function. As such, \indexshort is, at worst, equivalent to the entropy-based notions introduced by prior work, which focus on account balances alone. (We also show this empirically in Section~\ref{sec:vbe-in-practice}.)

\subsection{Implications of \indexshort}
\label{subsec:implications-of-vbe}
We now explain the theoretical insights implied by \indexshort.

\paragraph{\indexshort Master Theorem.} Our theoretical results all aim to show the impact of policy choices or system changes on DAO decentralization, in terms of \indexshort. They all have a similar structure: (1) we consider two systems such that the only difference between them is some ``transformation'' of interest, e.g., a portion of the voters become apathetic, votes are instead private, etc; (2) we reason about the impact of this transformation on the voting blocs of both systems; (3) based on this, we compute and compare the \indexshort of both systems. We now define a ``master'' theorem for \indexshort which captures this structure, and thus serves as a proof template that can be instantiated with concrete transformations of interest to prove different results. Our theorem is stated in terms of an arbitrary clustering function $C$, and \emph{min-entropy} as the entropy measure. Min-entropy captures the amount of ``information'' in the largest voting bloc by token holdings, i.e., for a set of sets of addresses $A$ with a total of $T$ tokens held across all individual accounts,
\small
\[F_{\min}(A, \tokens) \coloneqq \log_2\Bigg(\frac{\max\limits_{A' \in A} \tokens(A'
)}{T} \Bigg). \]
\normalsize 

We use min-entropy to state the VBE master theorem for clarity of presentation, but it can be refined to account for other entropy measures.

\begin{thm}[\indexlong Master Theorem]
\label{thm:vbe-master}
     We define $T$ to be a function that represents a \emph{system transformation}, i.e., a change in the players, elections, utilities of the players, and/or the distribution of tokens, which we denote by $(\players', E', U'_{E', \players}, \tokens') \coloneqq T(\players, E, U_{E, \players}, \tokens)$. The total number of tokens in the system stays constant, however. Let $B$ and $B'$ be the (not necessarily unique) largest clusters by token holdings according to clustering function $C$ for $(E, U_{E, \players}, \tokens)$ and $(E', U_{E', \players}, \tokens')$, respectively. Then, it follows that
     \small
    \begin{align*}
        \tokens'(B') &\geq \tokens(B) \iff \\ \indexshort_{C, \min}(E, \players, U_{E, \players}, \tokens) &\geq \indexshort_{C, \min}(E', \players', U'_{E', \players}, \tokens').
    \end{align*}
    \normalsize
\end{thm}
\begin{proof}
    This follows directly from the definition of $\indexshort_{C_{\epsilon}, \min}$:
    \small
        \begin{align*}
        \tokens'(B') \geq&\ \tokens(B) \\
        \iff  \frac{\tokens'(B')}{\sum\limits_{\player \in \players'} \tokens'(\player)} \geq&\ \frac{\tokens(B)}{\sum\limits_{\player \in \players} \tokens(\player)} \\
        \iff -\log_2 \Big(\frac{\tokens'(B')}{\sum\limits_{\player \in \players'} \tokens'(\player)}\Big) \leq &\ -\log_2 \Big(\frac{\tokens(B)}{\sum\limits_{\player \in \players} \tokens(\player)}\Big) \\
          \iff  \indexshort_{C, \min}&(E', \players', U'_{E', \players'}, \tokens') \\ \leq \indexshort_{C, \min}& (E, \players, U_{E, \players}, \tokens)
    \end{align*}
\end{proof}
\normalsize
Note that, if $B'$ represents a (new) majority by token holdings, then \indexshort strictly decreases; equality follows when $\tokens'(B') = \tokens(B)$.

This master theorem thus serves as a template that individual theorems can bootstrap from: simply specify a transformation $T$, explain how this modifies the largest voting bloc (if at all) according to the clustering function, and invoke Theorem~\ref{thm:vbe-master}. Armed with this formula, we now move on to concrete theoretical insights implied by \indexshort. Due to space constraints, we present the full theorem statement and proof for only the first of our results, to showcase the general structure of these. For the rest, we present the conceptual ideas behind the results here, and refer readers to Appendix~\ref{sec:A-additional-theorems} and the extended version of the paper for the full details,.

\newdiff{As explained in Section~\ref{subsec:ovbe-framework}, the \indexshort framework is compatible with any clustering function that defines a partition on $\players$, even if the function is vacuous or contrived. Therefore, in order to derive meaningful conclusions, our results will need to assume basic properties about the clustering algorithm in use, as generic results for arbitrary partitions would lead to trivial conclusions. For the rest of this section, we assume that we are dealing with any clustering function $C$ satisfying two simple properties. First, there exists some small constant $\epsilon$ such that, if $\player_i$ and $\player_j$ are such that $|\util_{\player_i}(e, \true)| \leq \epsilon$ and $|\util_{\player_j}(e, \true)| \leq \epsilon$ for all elections $e \in E$, then $C(U_{E, \player_i}, U_{E, \player_j}) = 1$. Second, if $\player_i$ and $\player_j$ are such that their utilities for every election have the same sign, then $C(U_{E, \player_i}, U_{E, \player_j}) = 1$. We note that our results can be modified in straightforward ways to accommodate for other ``natural'' assumptions of clustering functions, and thus our conceptual takeaways are general; we restrict ourselves to the aforementioned ones for clarity of presentation. More importantly, these assumptions are satisfied by all standard clustering functions if we focus our analysis on \emph{ordinal} utility functions. Looking ahead, this will be the case in Section~\ref{sec:vbe-in-practice}, where we empirically estimate \indexshort using voting history (there, we also discuss ordinal utilities, and their extensive use in economics, in more detail).

\paragraph{Result \#1: Owning multiple accounts.} As explained in Section~\ref{sec:introduction}, previous notions of entropy fail to capture the centralization that is present (but hidden) when a whale distributes tokens across multiple accounts / addresses. In such cases, it may appear that tokens are well diversified across accounts, while a large fraction are in fact under the control of one entity. Unlike prior notions, \indexshort captures this nuance: all these accounts would indeed be grouped together in a single voting bloc (we make the simplifying assumption that an individual's utility function is the same across all her accounts). We formalize this below.

\begin{thm}[Sybil Attacks and \indexshort]
\label{thm:owning-multiple-accounts}
     Let $(\players', E, U'_{E, \players'}, \tokens') = T_{\texttt{mult}}(\players, E, U_{E, \players}, \tokens)$ be the transformation where some player $\player \in \players$ divides her tokens across a new set of accounts $\hat{\players}$, i.e., $\players' = \players \cup \hat{\players}$, $\tokens'(\hat{\players}) = \tokens(\player)$, and $\forall \hat{\player} \in \hat{\players}$, $U'_{E, \hat{\player}} = U_{E, \player}$. The rest of the system remains unchanged. Then, it follows that
     \small
    \begin{equation*}
        \indexshort_{C, \min}(E, \players, U_{E, \players}, \tokens) = \indexshort_{C, \min}(E, \players', U'_{E, \players'}, \tokens').
    \end{equation*}
\end{thm}
\normalsize

\begin{proof}
Let $B$ be the largest voting bloc by token holdings before $T_{\texttt{mult}}$, which may or may not include $\player$. By assumption, all $\hat{\player} \in \hat{\players}$ are such that $U'_{E, \hat{\player}} = U_{E, \player}$. Thus, all new accounts will be in the same voting bloc $B'$ after $T_{\texttt{mult}}$, namely, $B' = [\player]$.

It follows then that, even though $\player$'s tokens are distributed between all individual accounts in $\hat{\player}$, they are in fact still under the control of the same block, i.e., $B'$. As such, $\tokens'(B') = \tokens(B')$. So, since no blocs acquire any new tokens, $B$ is still the largest voting bloc by token holdings after $T_{\texttt{mult}}$. Then, from Theorem~\ref{thm:vbe-master} it follows that
\small
\[\indexshort_{C_{\epsilon}, \min}(E, \players, U_{E, \players}, \tokens) = \indexshort_{C_{\epsilon}, \min}(E, \players, U'_{E, \players}, \tokens)\]
\normalsize
as desired.
\end{proof}
This result shows that, according to \indexshort, the ``true'' decentralization of the system does not change when a whale splits her tokens into multiple accounts, as they are all still under the control of the same voting entity. Conversely, metrics that focus on account balances alone would mistakenly conclude that the decentralization of the system strictly increased, since a set of tokens is diversified across more accounts.

\paragraph{Result \#2: Apathy.}  A system where voters are apathetic, i.e., not interested in the direction of the community, is not aligned with the goals of a DAO: distribution of tokens is irrelevant if individuals abstain from voting, as elections are narrowed squarely to the set of more invested stakeholders. Our definition captures this fact. Intuitively, \emph{apathetic voters all have similar utility functions}, which reflects their lack of stake in the elections. \indexshort groups all of these players within the same voting bloc, i.e., the cluster of voters with low utilities.

If the disinterested players are small stakeholders to begin with, apathy has a centralizing effect, as they now belong to a larger bloc of aligned voters. Indeed, in practice, it is common for the set of apathetic voters to represent a majority of token holdings~\cite{daoapathy,fritsch2022analyzing}. We refer to the bloc of apathetic voters in a DAO, i.e., non-voting token holders, as the \textit{inactivity whale}. This term reflects the collective and potentially systemically important inactive behavior of this group.

\paragraph{Result \#3: Delegation.} Intuition would suggest that delegation leads to a more centralized system: tokens that are originally held by a large set of players are instead under the control of the (smaller) set of delegates. However, \indexshort shows how this situation is more nuanced, as delegation actually tends to make a DAO \emph{more} decentralized: before delegation, the tokens are all held by a \emph{single} voting bloc, namely, the inactivity whale. Delegation then diversifies the tokens held by this ``whale'', and distributes them amongst a set of voting blocs (the delegates). Assuming that the size of the inactivity whale is larger than each delegate's total tokens---which tends to be true in practice~\cite{daoapathy,fritsch2022analyzing}---the system is now more decentralized. That is, as long as the delegates are not ``too big,'' delegation has a decentralizing effect. Conversely, if some delegate is a whale, or gets delegated an overwhelming majority of tokens, then the system may become more centralized. Thus, delegation is most useful in cases \iffull{where apathy is high}\else{of high apathy}\fi.

\paragraph{Result \#4: Herding.} A core goal of DAOs---and any democratic system more broadly---is for token holders to vote according to their true preferences. In practice, however, many DAOs exhibit \emph{herding} behavior: when votes are publicly observable, social dynamics lead to the formation of ``coalitions'' of voters. For example, token holders have reported feeling influenced to vote a certain way, often in alignment with influential community members, in order to thwart the reputational risks associated with opposing popular viewpoints~\cite{sharma2024unpacking}. Similarly, it has been observed and measured that token holders often vote in alignment with their peers~\cite{messias2023understanding}, who now operate as a single, large entity. In both cases, the utility derived from the social impact of a player's vote skews the utility of her desired outcome in a vacuum. Herding leads to more centralization, as votes artificially converge on one outcome. Unlike token distribution across individual accounts, \indexshort does conclude that reputational risk lead to more centralization: it aligns the utilities of the players towards the socially preferred outcome, which results in a bloc of aligned voters.

An important conclusion of this theorem is that privacy instead \emph{increases} the decentralization of a system, as it serves as a ``mitigation'' to herding. That is, if votes are private, token holders can vote for their true preferences, instead of being influenced by, e.g., social optics or the votes of their peers.

\paragraph{Result \#5: Voting slates.} Grouping together various elections into a lesser number of (more general) elections---so-called ``voting slates''---is in opposition with decentralization: decision-making is more diluted, thus decreasing the relative impact of each voter in the underlying proposals. That is, voting slates ``factor out'' differences in the viewpoints of individuals, yielding more homogeneous utility functions. For example, two players may disagree in many of the individual proposals, but agree on a few of the more important ones, resulting in them casting the same overall vote. \indexshort reflects the fact that bundling proposals indeed decreased decentralization: by considering a narrower set of elections, which smoothens utility functions, different voting blocs are combined to form larger ones.

\paragraph{Result \#6: Bribery.} There is an intuitive relationship between decentralization and bribery, namely, that successful bribery poses a threat to decentralization: in such a case, the entity that acquires the votes of the other players now commands a higher proportion of the total tokens than before. While ownership of tokens has not changed, \indexshort is sensitive to this centralizing effect, as it groups all bribed voters into the briber's bloc, since all bribee's now have aligned utility functions in line with the bribers desired outcome.

A second, more nuanced observation is that successful bribery must be \emph{systemic}, i.e., must involve a large number of tokens, if (and only if) a system is highly decentralized. Intuitively, if a DAO is highly centralized, a briber can directly coordinate with a few large players to guarantee an election outcome; or, if the briber is a whale herself, she only needs to bribe a few of the smaller players to accumulate enough tokens to mount a successful attack. Instead, in a more decentralized system, players are smaller, so a briber needs to widen the scale of their attack if they want to win an election. In this case, successful bribery requires large-scale coordination among various smallholders. Thus, as a DAO becomes more decentralized, a higher number of tokens need to be corrupted to guarantee an election outcome, since all players are small to begin with. Conversely, in a more centralized DAO, large whales only need to coordinate among themselves, or corrupt relatively few additional tokens to guarantee their desired election outcome.

Though a longstanding concern, systemic bribery is generally not considered realistic in secret ballot elections, due to logistical and economic challenges and criminal penalties---not to mention difficulty in verifying voters' compliance with bribers' demands. DAOs, however, are vulnerable: (1) Vote-buying is increasingly legal for proxy voting and thus may well be for DAOs~\cite{cole2001proxies}; and (2) Vote-buying in DAOs can be programmatically executed and verified by smart contracts, as shown by active marketplaces such as Votium~\cite{lloyd2023emergent}. It is even technically possible in principle for vote-buying to occur confidentially~\cite{philquadvotingdaos}.

\paragraph{Result \#7: Quadratic voting.} Quadratic voting~\cite{lalley2018quadratic} is a voting mechanism that attempts to dilute the influence of whales on election outcomes. To do so, a vote from a player that owns $n$ tokens will only have an impact of $\sqrt{n}$ in the outcome election. At face value, quadratic voting seems to make a system more decentralized: the quadratic ``tax'' is directly proportional to the number of tokens a player owns, which thus shrinks the gap between smaller players and whales. However, quadratic voting is known to be vulnerable to Sybil attacks and other forms of malicious 
coordination~\cite{weyl2017robustness}, and thus may have a \emph{centralizing} effect: players that are in large voting blocs implicitly subvert the quadratic tax due to the fact that their true token count is split among all bloc members. 
}

\section{\indexshort in Practice}
\label{sec:vbe-in-practice}
As we have emphasized throughout this work, utility functions are latent variables, and consequently \indexshort is as well. Thus, even though \indexshort serves as a powerful conceptual tool, it is not directly measurable, which is an inherent limitation of any metric that depends on
utility functions (including important results and models from voting theory, e.g.,~\cite{weyl2017robustness,duffy2008beliefs}). However, latent variables, such as utility functions, can be measured \emph{indirectly}, by inferring them via \emph{observable variables}, which do lend themselves to direct measurement. We introduce ``observable'' \indexshort (\oindexshort) in this section, a variant of \indexshort based on estimating utility functions from observable data. Unlike its latent counterpart, \oindexshort \emph{is} directly measurable, and thus provides an estimate for the ``true'' \indexshort of a DAO. \oindexshort represents the practical, measurable half of the \indexshort framework, which opens the door for exciting applications, many of which we highlight in this section.

\subsection{Observable \indexshort (\oindexshort)}
\label{subsec:ovbe}
Recall that the \indexshort framework consists of two main steps: (1) clustering players based on utility functions, followed by (2) computing some entropy measure over the distribution of tokens over the resulting clusters. \oindexshort simply adds a preliminary step to this template: represent each player's utility function in terms of relevant observable data, and perform the clustering based on this representation instead. We now define \oindexshort formally, and then move on to study \oindexshort in practice, including an empirical validation of \oindexshort.

\paragraph{Defining \oindexshort.}
\oindexshort extends the \indexshort framework by introducing a \emph{metric space} $M \coloneqq (S, d)$ as a third parameter (alongside the clustering function and entropy measure) to the framework, where $S$ is some set with $|\players|$ elements and $d$ is a distance function on $S$. The space $S$ represents a class of observable data that are being used to estimate utility functions, and the function $d$ allows a clustering metric to partition $S$, since clustering functions require some similarity metric between their inputs in order to assign clusters. We thus get the following definition of \oindexshort:

\begin{definition}[Observable \indexlong]
\label{def:govdec3}
For a set of elections $E$, a set of players $\players$, a mapping specifying the distribution of token ownership $\tokens$, a clustering function $C$, an entropy measure $F$, and a metric space $M = (S, d)$ with $|\players|$ elements, we define \textit{Observable \indexlong~(\oindexshort)} to be: 
\begin{equation*}
\indexshort_{C, F, M}(E,\players, \tokens) \coloneqq F(\players/\sim_{C_M}, \tokens).\end{equation*}
\end{definition}

In the notation above, $\players/\sim_{C_M}$ represents partitioning $\players$ based on how $C$ partitions $S$ using $d$. That is, $\player_i$ and $\player_j$ are in the same equivalence class if and only if $S_i \sim_{C_d} S_j$.

\paragraph{\oindexshort in practice.}
More intuitively, \oindexshort simply consists of (1) using some relevant dataset (e.g., on-chain metrics) that characterizes the DAO members' utility functions across elections, (2) defining a notion of ``distance'' between elements of this dataset, (3) computing a clustering function over this dataset (which uses the distance metric), and (4) computing the entropy metric over the tokens held by each cluster.

The observable data used for clustering is a key parameter for \oindexshort. There is a rich body of work dedicated to inference of latent variables (and of utility functions specifically)~\cite{messick1985estimating}, and indeed \oindexshort is agnostic to the method used. The most natural starting point, which we use throughout the rest of this work, consists of using \emph{voting history} as the observable data that estimate a player's utility function for a set of elections. Assuming that players are rational actors, their vote (or lack thereof) in an election is equivalent to the \emph{direction} of their utility for that particular election. That is, for any player $\player$ and election $e$, it follows that if $\vote_\player(e) \neq \bot$, then $\util_\player(e, \vote_\player(e)) > \epsilon$; conversely, if $\vote_\player(e) = \bot$, then $|\util_\player(e, \true)| < \epsilon$. As such, $\player$'s vector of votes in the set of elections $E$, denoted by $V_{E, \player}$, represents their \emph{ordinal} utility function for $E$.

In this example, the dataset for clustering consists of a matrix $\votehistory_{\players, E} \coloneqq (V_{E, \player})_{\player \in \players}$, where each row represents the list of votes for that particular player across elections. We can then use any number of distance metrics and clustering functions for vectors, such as cosine similarity or Euclidean distance for the  former, and K-means or hierarchical clustering for the latter. \newdiff{As with other applications of clustering, the best variant to use depends on the nature of the input data. So, the choice of algorithm is a function of the observable data being used, and its characteristics. Standard considerations for choosing between algorithms apply to our setting (e.g., dataset size, dimensionality of points, sparsity of data, etc), and so ample literature~\cite{xu2015comprehensive} on clustering functions is directly applicable. For the example of $\votehistory_{\players, E}$ as input data, many of these clustering considerations translate to the nature of the DAO's elections, such as the number of proposals and voters, the degree of voter apathy, etc.

We note that, in some cases, a DAO's proposals and goals may be purposefully designed for broad acceptance, e.g., proposals for security patches or defeating attacks. \oindexshort is equipped to capture this nuance, by using a clustering function that is refined to take this into account. For example, we can use a standard clustering function and apply less weight on these ``popular'' proposals, or remove these from the clustering altogether. Similarly, the duration of alignment for clustering can be adjusted, simply by modifying the size of each \oindexshort window, i.e., the number of proposals that are included in each computation of \oindexshort.}

How closely \oindexshort estimates \indexshort will be primarily based on the observable data that are gathered, and how accurately it estimates utility functions. We refer readers to external sources such as~\cite{messick1985estimating} for more details on this relationship, which is outside the scope of our work. As mentioned above, for the rest of this work we will use voting histories, i.e., ordinal utilities, to compute $\oindexshort$. Ordinal utilities are widely used in the economics literature as a suitable estimate for utility functions~\cite{manski1988ordinal} (and, in fact, are equivalent to cardinal utilities in some cases~\cite{vnm-utility:2024}), and thus our instantiation of \oindexshort serves as a good proxy for \indexshort. \newdiff{Other potential on-chain data that could indicate alignment include, for example, assets owned, membership in other DAOs, sponsored proposals, etc., and thus could potentially also be used as observable data.} An interesting direction for future work is the use of off-chain data sources, such as social media interactions or low-cost straw polls, to refine estimates of utility functions.

\paragraph{\oindexshort open-source toolkit.} Now that we have introduced \oindexshort, we are ready to perform measurement studies and derive \indexshort-based applications. For this purpose, we prepared a comprehensive \oindexshort toolkit that can be used for easily computing \oindexshort as part of broader applications. Our toolkit, which we make available to the community as an open-source artifact~\cite{vbezenodo}, contains a variety of popular distance metrics, clustering functions, and entropy metrics, which can be assembled together to get a myriad of \oindexshort instantiations. Our toolkit can be used as both a standalone program or integrated as a library in applications, and simply requires users to input a dataset and select a clustering function, distance function, and entropy measure. 
\newdiff{In the documentation of our toolkit, we include guidance for selecting specific \oindexshort instantiations based on the nature of the input data, to complement existing documentation on choices of algorithms.}

Our toolkit is modular by design and can be easily extended to add support for more algorithms. Further, our toolkit provides a number of additional features besides just computing \oindexshort. For example, it has support for an intermediate report of the clusters and support for automatic detection of parameters. We refer readers to the documentation of our toolkit for more details on the features included.

While our toolkit is agnostic to the dataset used as input, we assume that voting history will often be used as observable data. As such, we include scripts for scraping the voting history of DAOs using open source APIs in a format that is directly compatible with the rest of our toolkit. (As described in Section~\ref{subsec:ovbe-landscape}, scraping voting history is surprisingly difficult).

In the rest of this section, we put this toolkit to the test by performing a variety of \oindexshort-based empirical studies. While these are independent contributions, they additional serve to showcase the flexibility and uses of our toolkit (and of \oindexshort more broadly).

\subsection{Empirical Validation of \oindexshort}
The starting point for our empirical work is an an experimental validation of \oindexshort. We first identify a \emph{natural experiment} that can be used as the basis for testing a decentralization metric, which we then apply to \oindexshort.

\paragraph{A natural decentralization experiment.} An ideal experiment to validate a decentralization metric consists of a single population of players who vote on two sets of proposals, such that the proposals in one set are purposefully crafted to yield more centralized results than the other. Then, one could independently compute a candidate decentralization metric on both sets of proposals, and see whether the outputs agree with the expected centralization.

It turns out that a \emph{natural} example of this experiment takes place in some DAOs today, in the form of \emph{two-round voting}. In these DAOs, proposals initially go through a preliminary round of voting---a so-called \emph{temperature check}---during which the community discusses the proposal and agrees on its parameters. This is followed by a second round of voting, which decides whether the proposal is ultimately accepted. The goal of the first voting round is precisely to form consensus around proposals, thus \emph{explicitly inducing a form of centralization for the second round}. As such, two-round voting is, by design, a clean, natural playground to test a decentralization metric: we can independently compute the metric on (1) all off-chain, temperature-check proposals, and (2) all on-chain, second round proposals, and confirm whether the metric reports higher centralization for the second type.

Of all DAOs with two-round voting, \emph{Uniswap} serves as a particularly compelling test case for this experiment. First, as a mature DAO with years of governance history, Uniswap has one of the largest DAO treasuries and readily-available proposal history \cite{DeepDAO:2023}. Second, Uniswap follows a straightfoward, unicameral legislative structure with UNI token holders serving as the sole decision-making body for all proposals \cite{uniswap-governance-docs}. This is unlike other major DAOs that follow a bicameral process, such as Arbitrum's Security Council~\cite{ArbitrumDAO:2023} and Optimism's Citizen House \cite{optimism-governance-process}.

Uniswap also explicitly defines the purpose of off-chain temperature checks to signal community sentiment prior to an on-chain vote, and the on-chain vote should incorporate feedback from prior voting rounds \cite{uniswap-governance-docs, uniswap-gov-forum}. Because of this governance mechanism design, we can reasonably infer that there would be more agreement in the second-round on-chain vote versus the off-chain vote.

\paragraph{\oindexshort validation.} We performed this experiment on Uniswap using \oindexshort as a decentralization metric, and confirm that \oindexshort \emph{reports a higher level of decentralization in the first round of voting}. We computed \oindexshort using K-means for clustering and both min-entropy and Shannon entropy, in windows of 10 proposals across Uniswap data on Snapshot (round 1) and Tally (round 2). We then averaged these values to obtain the average \oindexshort for each voting round. The first round of voting resulted in an \oindexshort of 0.7804 (min-entropy) and 1.3149 (Shannon entropy), while round two resulted in 0.6601 (min-entropy) and 1.1527 (Shannon entropy).

Importantly, note that, unlike \oindexshort, prior decentralization metrics would \emph{not} pass this empirical validation test, since the centralization induced by the first round of voting is orthogonal to account balances. In particular, if all token balances stay the same, these metrics would report that decentralization did not change.

\subsection{The \oindexshort Landscape}\label{subsec:ovbe-landscape}
\begin{figure*}[t]
    \centering
    \footnotesize
    \renewcommand{\arraystretch}{1.2}
    \begin{tabular}{|p{1.4cm}|p{1cm}|p{1.1cm}|r|p{1.8cm}|p{1.25cm}|p{1.2cm}|p{1.2cm}|p{1.2cm}|p{1.5cm}|}
    \hline
    \textbf{DAO} & \textbf{Avg. VBE} & \textbf{Std. dev. of VBE} & \textbf{Treasury} & \textbf{Category} & \textbf{Unique Voters} & \textbf{Total Proposals} & \textbf{Avg Voter Participation} & \textbf{Nakamoto Coefficient} & \textbf{Gini Index} \\
    \hline
    Optimism & 0.9929 & 0.2461 & \$2,800.00M & Infrastructure & 221,066 & 222 & 20.18\% & 20 & 0.997135224 \\
    Arbitrum & 0.9254 & 0.2374 & \$2,100.00M & Infrastructure & 283,300 & 544 & 19.18\% & 24 & 0.994582646 \\
    Nouns DAO & 0.8806 & 0.2493 & \$13.30M & DAO Tool & 1,059 & 563 & 2.97\% & 35 & 0.779029699 \\
    Gitcoin & 0.8536 & 0.2417 & \$42.20M & Funding & 11,903 & 141 & 9.87\% & 16 & 0.994281210 \\
    Uniswap & 0.7435 & 0.2161 & \$2,500.00M & DeFi & 68,240 & 328 & 9.56\% & 28 & 0.998941837 \\
    Decentraland & 0.5995 & 0.2544 & \$93.80M & Gaming, NFTs & 1,244 & 285 & 4.62\% & 6 & 0.972266008 \\
    Aave & 0.5394 & 0.2320 & \$138.90M & DeFi & 9,565 & 398 & 0.70\% & 9 & 0.995135576 \\
    \hline
    \end{tabular}
\captionof{table}{Summary of a subset of the DAOs included in our measurement study of \oindexshort calculated with min-entropy.}
\vspace{-0.3cm}
\label{fig:vbe-dashboard-summary}
\end{figure*}

After our empirical validation of \oindexshort, our next goal was to conduct a measurement study of DAOs, to understand how decentralization compares across the DAO ecosystem. We discuss the main steps of this process below. For a more extensive account of our methodology, we refer readers to the documentation in our open-source \oindexshort artifact ~\cite{vbezenodo}.

\paragraph{Step \#0: corpus of DAOs.} We first assembled a corpus of target DAOs to include in our analysis. We defined qualifying DAOs for the study as those with the top 20 largest treasury sizes, at least 25 recorded proposals, over 5,000 unique voters, and data availability in open source APIs for Snapshot and Tally. Narrowing down this list was critical for resource management, data quality, and applicability of the study. The final list included 34 DAOs.

\paragraph{Step \#1: data collection.} Once we had a list of target DAOs, the next step consisted of scraping the voting history and token holdings of all members of each DAO, to use as inputs to \oindexshort. Acquiring this data turned out to be a surprisingly difficult task. As prior work also points out, ``in practice it is not trivial to acquire all governance related information from raw blockchain data''~\cite{feichtinger2023hidden}. The greatest challenge we faced in this step was navigating the heterogeneity of voting data formats across DAOs, which required us to create data extraction scripts to assemble a full dataset of voting history and incorporate manually added metadata; we make these scripts available as part of our \oindexshort toolkit. \newdiff{Collecting the governance data of all 34 DAOs from our study took approximately 8 hours using a typical laptop.}

\paragraph{Step \#2: computing \oindexshort.} Once we had a dataset with the voting history of each DAO, we used our toolkit to compute \oindexshort across proposal windows. We used K-means with $k = 3$ as our clustering function, since align with the nature of voting choices which are normally ``Yes,'' ``No,'' and ``Abstain.'' In addition, we used Euclidean distance as the distance metric for clustering, a rolling window of 10 proposals, and min-entropy and Shannon entropy as our entropy measures. For comparison, we additionally computed other popular decentralization metrics for DAOs, such as the Gini Nakamoto coefficients. \newdiff{Performing all these computations took approximately 3.5 hours using a typical laptop. In our study, as little as 10 proposals sufficed for clustering. Exploring the theoretically-minimum number of proposals required to compute oVBE is an interesting direction for future work.}

\paragraph{Step \#3: results and analysis.} We display the results of our measurement study in a public dashboard, which can be found at~\cite{daovbedashboard}. We also show a summary of a few DAOs in Table~\ref{fig:vbe-dashboard-summary}. Our dashboard updates weekly to show the evolving \indexshort landscape. The focal point of our dashboard is the \oindexshort of each DAO. The dashboard first presents an overview of all DAOs, with high-level statistics such as proposal count, unique voters, voter participation, \oindexshort, Gini and Nakamoto coefficient. The DAO-specific view displays current, average, max, and min \oindexshort, as well as a graph of \oindexshort fluctuations. Finally, the proposal view specifies vote counts and voting power allocated to proposal choices, discussion URL, description, and outcomes by voter power and by vote count.

\paragraph{Dashboard use cases.} The main use case of our dashboard is to track and monitor the \oindexshort of a particular DAO and understand shifts in decentralization over time. This can help reason about the overall ``health'' of a DAO, and investigate potential forces in proposals or voter behavior that may change \oindexshort internally in an organization. For example, a stakeholder may want to understand the cause of a spike or drop in \oindexshort and whether it aligns with their expectations. These metrics and understanding of an organization can be used to improve a DAO's governance structure, build safeguards to protect against governance attacks, or serve as an alert function for risks or suspicious behavior. We note, however, that a limitation of \oindexshort is that, a priori, it does not provide direct insights on \emph{cross-DAO} comparisons.

As a second example, our dashboard can also help practitioners refine delegation measures and governance mechanisms. Clustered voting blocs can allow voters to easily identify delegates that align with their preferences, reducing this barrier of entry to delegation and thus increase voter participation in DAOs. Furthermore, many DAOs today implement measures to increase decentralization by redistributing voting power, such as delegating large amounts of voting power to a council of independent delegates \cite{optimism-anti-capture}. DAOs can dynamically change the amount of voting power they provide to these councils in order to stay within a target range of \oindexshort.

\subsection{\oindexshort in Governance Experiments}\label{subsec:ovbe-gov-experiments}
We discuss another application of \oindexshort in this section: as a decentralization metric in governance experiments. \oindexshort can be a useful metric to, for example, empirically test hypotheses about the impact of mechanism choices on the decentralization of a DAO. \newdiff{In particular, such governance experiments can be used to, e.g., empirically verify the impact of mechanisms designed to increase the security of a DAO by increasing decentralization (Section~\ref{subsec:vbe-security}).} While governance experiments can use any decentralization metric as their outcome, \oindexshort has the potential to be particularly informative due to its sensitivity to social dynamics like herding and apathy. We describe this application through the lens of a real-world case study, in the form of an ongoing collaboration with the Optimism Foundation.

\mypara{Example: collaboration with Optimism.} Optimism~\cite{Optimism:2023} is a collective that runs the popular RetroPGF~\cite{OptimismPGF} program, an initiative based on the idea of \emph{retroactively rewarding} projects that provide positive value to the community. Every funding round, a set of chosen voters determine the ``impact'' of each nominated project, by submitting ballots in the form of a distribution of some total dollar amount across the set of nominated projects. These ballots are then aggregated to determine the monetary reward for each project. To date, Optimism has had four rounds of funding, allocating over \$100 million to various projects~\cite{retropgf-funds}.

Maintaining a healthy level of decentralization is of central importance for RetroPGF and its security, as collusion between voters could lead to millions of dollars of misallocated funds. For this reason, Optimism is performing an experiment on a future round of RetroPGF, to understand whether changes to their governance structure would lead to higher decentralization. In particular, their goal is to test whether a \emph{random selection of voters} (sampled from the Optimism ecosystem) leads to a more decentralized funding round than \emph{web of trust}, which is their current mechanism to select voters. This hypothesis will be tested by comparing the casted ballots of two groups---a set of 100 voters selected at random, and a set of 100 voters selected via web of trust---using some decentralization metric to quantify the alignment of each group. This experiment would help inform whether random sampling could lead to a more heterogeneous group of voters.

Optimism has decided to use \oindexshort as their decentralization metric to use for this upcoming experiment. We worked with them to adapt \oindexshort to their particular governance structure, showcasing \oindexshort's ability to mold itself to a variety of governance structures. The main challenge in this process was selecting an instantiation of \oindexshort that matched their non-standard voting structure (i.e., where ballots are vectors of integers instead of a single, binary number). To address this, we adapted clustering functions used in other domains, where inputs are also vectors of higher dimensions, such as document clustering tasks~\cite{huang2008similarity}. This is based on the insight that, conceptually, we can think of a retroPGF ballot as analogous to a \emph{document in the bag-of-words model}, and thus we can leverage document similarity metrics for the clustering of voters. At a high-level, these metrics tend to follow a typical template, where documents are first pre-processed to normalize for common words, followed by a clustering algorithm (such as K-means or hierarchical clustering) on the input vectors, which under-the-hood relies on a distance metric for the vectors, e.g., Euclidean distance or cosine similarity. In our case, the pre-processing step would correspond to normalizing for highly popular projects.

The result of Optimism's experiment is pending. However, our process of working with them serves as a case-study of \oindexshort's utility as a metric for governance experiments, and its flexibility to meet a wide range of requirements and unique settings.

\section{Summary Guidance for DAOs}\label{sec:guidance}
Our \indexshort framework and the theoretical implications we show in~\Cref{sec:vbe-in-theory} suggest a number of forms of concrete guidance for DAOs seeking to enforce or improve meaningful decentralization. We summarize our guidance for practitioners in~\Cref{tab:recommendations}.

\begin{figure}[h!]
    \centering
    \renewcommand{\arraystretch}{0.6}
\resizebox{\columnwidth}{!}{
\LARGE
\begin{tabular}{|p{3cm}|p{7.5cm}|p{9cm}|p{3.45cm}|}
  \hline     \textit{Topic}  & \textit{General Guidance} & \textit{Reason}  & \textit{Relevant result}
       \\
        \hline &&&\\
     1. \textbf{Vote delegation} & Given a large inactivity whale, vote delegation tends to increase decentralization.  & Delegation \iffull (counterintuitively) \fi increases decentralization  by diversifying tokens away from a big inactivity whale.  & Thm.~\ref{thm:delegation}\\
        \hline &&&\\
     2. \textbf{Voting privacy} & Voting privacy increases decentralization. & Private voting eases herding, whose effects are centralizing.  & Thm.~\ref{thm:herding}\\ 
             \hline &&&\\
      3. \textbf{Voter bribery} & The scale of bribery increases with decentralization. & Low alignment of utility functions means systemic coordination is required to impose alignment.   & Thms.~\ref{thm:bribery1},~\ref{thm:bribery2}, and~\ref{thm:bribery3}\\
              \hline &&&\\
    4. \textbf{Identity verification} & Weak identity verification increases centralization in quadratic voting. &  A whale that can spread tokens across identities amplifies its voting power.  & Analysis in Section~\ref{subsec:implications-of-vbe}\\ 
                 \hline &&& \\
   5. \textbf{Voting slates / proposal bundling} &  Bundling choices into slates (like protocol upgrades that include many voting issues in one package) decreases decentralization. & Bundled choices artificially align otherwise heterogeneous utility functions and/or induce apathy by smoothing out utility functions.  & Thm.~\ref{thm:voting-slates}\\ 
             \hline &&&\\
   6. \textbf{Data collection} & Careful voting-statistic collection facilitates decentralization measurement. & Lack of systematic collection and publication of detailed voting statistics makes decentralization measurement challenging today.     & Discussion in~\Cref{sec:vbe-in-practice}.\\ 
                \hline 
    \end{tabular}
}
    \captionof{table}{Guidance implied by this paper's results regarding DAO decentralization.}
     \label{tab:recommendations}
\end{figure}

\paragraph{Apathy / inactivity whale and delegation.} 
One way to diminish the size of the inactivity whale is through delegation. Intuitively, if tokens associated with the inactivity whale are distributed between at least two delegatees in distinct clusters, then they come to represent distinct utility functions—--and thus contribute to decentralization. Our theorems also show that when the inactivity whale is large---with respect to delegatees---delegation increases decentralization. (Otherwise, delegation may or may not have this effect.)

\paragraph{Herding / voting privacy.} Herding arises because votes are publicly visible. Voting privacy in principle alleviates such pressure and therefore has a decentralizing effect. Snapshot, a popular platform for DAO voting, has recently implemented a form of privacy called \textit{shielded voting}~\cite{Snapshot:2022}. This form of privacy, however, is only ephemeral: Votes are private when submitted, but revealed at the end of the vote-casting period. So it is unclear that it can fully address the centralizing effects of herding. End-to-end verifiable voting systems have been proposed in the literature for decades that achieve both voting integrity and confidentiality~\cite{ali2016overview}. How to implemented them with token-based weighting is, to the best of our knowledge, though, an open problem.

\paragraph{Voter bribery.} DAOs today are largely 
centralized~\cite{sharma2024unpacking, feichtinger2023hidden, bribery}. Bribery may not be especially useful, as whales generally exert strong control and require relatively little coordination to align utility functions into a favorable voting bloc. Voter bribery, however, is a problem in many settings, both in political voting~\cite{mcgee2023often} and in corporate governance (see, e.g.,~\cite{Schickler:2023}). One implication of our results is that as DAO decentralization increases, in order for bribery to succeed, it will need to be systemic. DAO designers should therefore recognize large-scale bribery as a future risk.

\paragraph{Voting slates / bundling proposals.} As the practice of bundling proposals / measures has the goal of aligning utility functions, 
from the standpoint of \indexshort, it generally has a centralizing effect. DAOs may therefore wish to consider limiting the practice and instead explore way to unbundle multi-component proposals.

\paragraph{Data collection.}
As discussed in Section~\ref{sec:vbe-in-practice}, we found it challenging to collect full voting histories even for popular DAOs. A recommendation for the community is to establish and adhere to standards for archiving DAO voting data.

\newdiff{
\paragraph{Future work: practical insights from LHD.} In Section~\ref{sec:rl-decentralization}, we introduce the \hypothesisname (\hypothesisnameshort), a conceptual model for DAO decentralization, which we use to derive \indexshort. Yet, we believe \hypothesisnameshort has profound (practical) applications to DAO decentralization 
beyond just \indexshort. For example, techniques to increase diversity in MARL could form the basis for the development of mechanisms aimed at increasing DAO decentralization. Similarly, LHD could be used to motivate existing approaches to decentralization that otherwise lack principled motivation.
For example, existing proposals for rewarding the use of voting power (by voting or delegating) is akin to diversity-boosting in MARL via tailoring of reward functions. Exploring these (and other applications of \hypothesisnameshort) is an interesting direction for future work.}

\section{Related Work}
\label{sec:related-work}

\paragraph{DAOs.} Research literature on DAOs has been limited to date, but has included measurement studies~\cite{feichtinger2023hidden,sharma2024unpacking}, retrospectives on the failure of The DAO (e.g., \cite{dupont2017experiments}) and ways of addressing related technical flaws in smart contracts such as dangerous reentrancy 
(e.g.,~\cite{luu2016making}), DAO mechanism design (e.g.,~\cite{bahrani2023bidders}), and exploration of DAOs from the standpoint of legal theory 
(e.g.,~\cite{hassan2021decentralized}) and economics and governance (e.g.,~\cite{beck2018governance}). Works exploring measurement of DAOs' degree of decentralization most notably include Feichtinger et al.~\cite{feichtinger2023hidden}, who explore Gini and Nakamoto indices, as well as participation rates and the monetary cost of governance, Sharma et al.~\cite{sharma2024unpacking}, who consider various notions of entropy, participation rate, and graph-based measures of decentralization, and~\cite{wright2021measuring}, which taxonomizes DAOs by comparison with other autonomous systems. Sun et al.~use clustering to identify voting blocs in a study of MakerDAO~\cite{sun2022multiparty}. Also of note is the informal notion of ``credible neutrality,'' a community standard articulated in, e.g.,~\cite{Buterin:2020,Buterin:2023}.

\paragraph{Social choice and voting theory.}
A long line of work on social choice and voting theory investigates how best to aggregate preferences of individual voters (see, 
e.g,~\cite{fishburn2015theory,kelly2013social} for overviews)---the same functionality that DAOs seek to provide in the decentralized setting. There are some major differences between the classical and DAO settings, however. For instance, the permissionless nature of DAOs---and use of token weighting---changes the nature and meaning of voter participation. Further, while the threat of large-scale voter bribery is typically safe to ignore in classical voting, both due to the high likelihood of detecting such an attack, as well as the the challenge in coordinating the attack itself, it is a realistic threat in DAOs, as we have explained. 

\section{Conclusion}
\label{sec:conclusion}

We have introduced \indexlong (\indexshort), a metric for DAO decentralization stemming from a new model of DAOs that derives from multi-agent reinforcement learning (MARL). \indexshort’s sensitivity to aligned behavior among voters makes it a powerful tool, yielding insights well beyond those of existing token-distribution-based metrics. \indexshort yields both theoretical results and practical guidance for DAO communities, particularly through its observable variant \oindexshort. We envision that the open-source dashboard and analysis tools we introduce here will serve as springboards for principled future research and practices that enhance the governance efficacy and security of DAOs. 

\section*{Acknowledgments}

This work was funded by NSF CNS-2112751 and CNS-2112726, a Sui Academic Research Award, a Uniswap Foundation TLDR fellowship, and generous support from IC3 industry partners and sponsors. We also wish to thank the Optimism Foundation---especially Eliza Oak and Thomas Bialek---for helpful guidance and collaboration.

\bibliographystyle{plain}
\bibliography{main}

\appendix

\section{Additional Theorems and Proofs}\label{sec:A-additional-theorems}
In this section, we show the rest of our theorems and proofs from the theoretical implications of \indexshort (Section~\ref{sec:vbe-in-theory}). Due to space constraints, we defer some of our proofs to the extended version of the paper.

\begin{thm}[Apathy and \indexshort]
\label{thm:apathy}
    Let $(E, U'_{E, \players}, \tokens) = T_{\texttt{apath}}(\players, E, U_{E, \players}, \tokens)$ be the transformation where players $\hat{\players} \subseteq \players$ become $\epsilon$-apathetic, i.e., $\forall \player \in \hat{\players}$ $\forall e \in E,$ $|\util'_\player(e, \true)| \leq \epsilon$. The rest of the system remains unchanged. Then, if $\forall \player \in \hat{\players},$ $\tokens(\apath) \geq \tokens([\player])$, it follows that
    \begin{equation*}
        \indexshort_{C_{\epsilon}, \min}(E, \players, U_{E, \players}, \tokens) \geq \indexshort_{C_{\epsilon}, \min}(E, \players, U'_{E, \players}, \tokens).
    \end{equation*}
\end{thm}

\begin{proof}
Let $B$ be the largest voting bloc by token holdings before $T_{\texttt{apath}}$. We first note that all apathetic voters belong to the same voting bloc $B'$, according to $C$: by the definition of $\epsilon$-apathetic, it follows that, for all $\player_i, \player_j \in \hat{\players}$ and $e \in E$,
\[|\util'_{\player_i}(e, \true)|, |\util'_{\player_j}(e, \true)| \leq \epsilon,\]

which corresponds precisely to the bloc of apathetic voters in $C$, containing all players in $\apath$. Then, by assumption, $\tokens(B') = \tokens(\apath) \geq \tokens([\player])$, $\forall \player \in \hat{\players}$. So, since no other blocs decrease in size, it follows that $\tokens(B') \geq \tokens(B)$: either the bloc that aggregates all apathetic voters is now the largest bloc, or the same bloc is the largest in both instances. Thus, from Theorem~\ref{thm:vbe-master}, it follows that
\small
\[\indexshort_{C_{\epsilon}, \min}(E, \players, U_{E, \players}, \tokens) \geq \indexshort_{C_{\epsilon}, \min}(E, \players, U'_{E, \players}, \tokens)\]
\normalsize
as desired.
\end{proof}
\normalsize

\begin{thm}[Delegation and \indexshort]
\label{thm:delegation}
Let $(E, U'_{E, \players}, \tokens') = T_{\texttt{deleg}}(\players, E, U_{E, \players}, \tokens)$ be the transformation where players $\hat{\players} \subseteq \players$, who are $\epsilon$-apathetic, instead delegate their votes to a set of delegates $D \subset \players$, i.e., $\tokens'(D) = \tokens(D) + \tokens(\hat{P})$ and $\tokens'(\hat{\players}) = 0$. The rest of the system remains unchanged. Then, if $\forall d \in D,$ $ \tokens(\apath) \geq \tokens'([d])$, it follows that,
\small
    \begin{equation*}
        \indexshort_{C_{\epsilon}, \min}(E, \players, U'_{E, \players}, \tokens') \geq \indexshort_{C_{\epsilon}, \min}(E, \players, U_{E, \players}, \tokens).
    \end{equation*}
\end{thm}
\normalsize

\begin{proof}
    Let $B$ by the largest voting bloc by token holdings before $T_{\texttt{deleg}}$. As discussed in the proof of Theorem~\ref{thm:apathy}, all players in $\hat{P}$ belong to the same voting bloc for all elections in $E$---the inactivity whale---since they are all part of the set of apathetic voters $\apath$. Let $B'$ be the largest voting bloc by token holdings after $T_{\texttt{deleg}}$; note that it may be the case that $B' = [d]$ for some $d \in D$.

    We first note that $B'$ is equal to either (1) $B$ itself, (2) the second largest voting bloc after $B$ before delegation, or (3) $[d]$, for some $d \in D$. That is, since the only blocs that change after $T_{\texttt{deleg}}$ are all the $[d]$ and the inactivity whale (which lost $\tokens(\hat{P})$ tokens), it must the be case that the new largest voting bloc is either the same one as before delegation, the second largest voting bloc before delegation (i.e., $B$ was the inactivity whale, which got fractionated by delegation), or one of the $[d]$ which increased in size.
    
    For (1) and (2), it is clearly the case that $\tokens(B) \geq \tokens'(B')$. Then, for (3), note that, by assumption, $\tokens(\apath) \geq \tokens'([d])$, for all $d \in D$. So, $\tokens(B) \geq \tokens(\apath) \implies \tokens(B) \geq \tokens'([d]) = \tokens'(B')$. 
    
    It follows then that, in all cases, $\tokens(B) \geq \tokens'(B')$. Thus, from Theorem~\ref{thm:vbe-master}, we get that
    \small
    \[\indexshort_{C_{\epsilon}, \min}(E, \players, U'_{E, \players}, \tokens') \geq \indexshort_{C_{\epsilon}, \min}(E, \players, U_{E, \players}, \tokens)\]
    \normalsize
    as desired.
\end{proof}

\begin{thm}[Herding and \indexshort]
\label{thm:herding}
    Let $(E, U'_{E, \players}, \tokens) = T_{\texttt{herd}}(\players, E, U_{E, \players}, \tokens)$ be the transformation where players $\hat{\players} \subseteq \players$ exhibit herding toward, without loss of generality, \true. That is, for all $P \in \hat{P}$ and $e \in E$, the monetary reputational cost of voting for \false\ is greater than or equal to
    $\max(2 \cdot \util_{\player}(e,\false]) + \epsilon, 0)$ for some constant $\epsilon$. The rest of the system remains unchanged. Then, it follows that
    
    \begin{equation*}
        \indexshort_{C_{\epsilon}, \min}(E, \players, U_{E, \players}, \tokens) \geq \indexshort_{C_{\epsilon}, \min}(E, \players, U'_{E, \players}, 
        \tokens).
    \end{equation*}
\end{thm}

\begin{proof}
Let $B$ be the largest voting bloc by token holdings before $T_{\texttt{herd}}$. Note that, after $T_{\texttt{herd}}$, all voters in $\hat{P}$ belong to the same voting bloc $B'$: for every $\player \in \hat{\players}$, $U'_{E, \player}$ will consist of only positive values: either $\player$ preferred an outcome of \true\ in $e$ to begin with, or their monetary utility of \true\ is now $|\util_\player(e, \false)| + \epsilon$. Thus, since $\sgn{\util_\player(e, \true)} = 1$ for all $e \in E$, all of $\hat{P}$ consists of a single voting bloc $B'$ according to $C$. 

It follows then that $\tokens(B') \geq \tokens(B)$, as either the ``new'' voting bloc $B'$ is now the largest bloc, or the same bloc is the largest before and after $T_{\texttt{mirr}}$. Then, from Theorem~\ref{thm:vbe-master}, it follows that
\[\indexshort_{C_{\epsilon}, \min}(E, \players, U_{E, \players}, \tokens) \geq \indexshort_{C_{\epsilon}, \min}(E, \players, U'_{E, \players}, \tokens)\]
as desired. 
\end{proof}

Note that the corollary that privacy \emph{increases} decentralization follows directly via a proof by contradiction of Theorem~\ref{thm:herding}).

\paragraph{Result \#5: voting slates.} We model a player's utility for a slate of elections simply by adding the utilities of the underlying proposals. That is, for all $\player \in \players$ and some election $\mathcal{E}$ comprised of some subset of elections of $E$, the utility of $\player$ in $\mathcal{E}$ is:

\[\util_\player(\mathcal{E}, \true) = \sum\limits_{e \in \mathcal{E}} \util_\player(e, \true).\]

Voting slates are generally used to ``hide'' unpopular proposals among a larger set of benign, popular proposals, and thus increase their chances of passing. We model this by saying that if two $\player_i, \player_j$ have aligned utilities (according to $C$) on all proposals underlying $\mathcal{E}$, then they will agree on $\mathcal{E}$ itself, i.e., $C_\epsilon(U_{E, \player_i}, U_{E, \player_j}) = 1$$\implies \sgn{\sum\limits_{e \in \mathcal{E}} \util_{\player_i}(e, \true)} = \sgn{\sum\limits_{e \in \mathcal{E}} \util_{\player_j}(e, \true)}$.

\begin{thm}[Voting Slates and \indexshort]
\label{thm:voting-slates}
    Let $(E', U'_{E', \players}, \tokens) = T_{\texttt{slates}}(\players, E, U_{E, \players}, \tokens)$ be the transformation where all elections $E$ are bundled together into slates to form a smaller set of elections $E'$. The rest of the system remains unchanged. Then, it follows that
    \small
    \begin{equation*}
        \indexshort_{C_{\epsilon}, \min}(E, \players, U_{E, \players}, \tokens) \geq \indexshort_{C_{\epsilon}, \min}(E', \players, U'_{E', \players}, \tokens).
    \end{equation*}
    \normalsize
\end{thm}

\begin{proof}
Let $B$ be the largest voting bloc by token holdings before $T_{\texttt{slates}}$. Then, note that all players in $B$ are still in the same voting bloc $B'$ after $T_{\texttt{slates}}$: since $C_\epsilon(U_{E, \player_i}, U_{E, \player_j}) = 1$ for every pair of players in $B$, by assumption, it follows that \[\forall \mathcal{E} \in E',\ \sgn{\sum\limits_{e \in \mathcal{E}} \util_{\player_i}(e, \true)} = \sgn{\sum\limits_{e \in \mathcal{E}} \util_{\player_j}(e, \true)}.\]

Conversely, players who did not belong to $B$ may, in fact, join $B'$ after $T_{\texttt{slates}}$: even if the players disagree in some of the underlying proposals for a particular slate $\mathcal{E}$, they may cast the same overall vote for the entire slate. As such, $B'$ contains strictly more players than $B$, which implies that $\tokens(B') \geq \tokens(B)$. Then, from Theorem~\ref{thm:vbe-master}, it follows that
\[\indexshort_{C_{\epsilon}, \min}(E, \players, U_{E, \players}, \tokens) \geq \indexshort_{C_{\epsilon}, \min}(E', \players, U'_{E', \players}, \tokens)\]
\end{proof}

\paragraph{Result \#7: bribery.} Before presenting our results, we first introduce some additional notation. We define $\bribe_\player \colon E \times \{\true, \false\} \to \mathbb{R}$ to be such that it is possible for player $P$ to achieve an outcome of $\true$ (resp., $\false$) in a particular election $e$ via bribery for any expenditure greater than $\bribe_\player(e, \true)$ (resp., $\bribe_\player(e, \false)$). Note that we make the simplifying assumption that bribery costs are independent across elections. We assume that bribing a given $\player$ to flip its vote from $\true$ to $\false$ (respectively, $\false$ to $\true$) costs $\max(2\cdot\util_{\player}(e,\true) + \epsilon, 0)$ (respectively, $\max(2 \cdot \util_{\player}(e,\false) + \epsilon,0)$), for some constant $\epsilon$. Successful bribery to achieve an outcome of $\true$ in $e$ means flipping enough votes to cross a certain threshold $q$ of votes for $\true$ in $e$ (and vice versa for $\false$). For example, a typical value for $q$ may be $q = 0.5$. We note that this represents the threshold to \emph{ensure} the desired outcome in an election, and not just to win it.

\begin{thm}[Bribery and \indexshort]
\label{thm:bribery1}
    Let $(E, U'_{E, \players}, \tokens) = T_{\texttt{bribe}}(\players, E, U_{E, \players}, \tokens)$ be the transformation where an entity successfully bribes players $\hat{\players} \subseteq \players$ in elections $E$ to achieve an outcome of, without loss of generality, $\true$. The rest of the system remains unchanged. Then, it follows that
    
    \begin{equation}
        \indexshort_{C_{\epsilon}, \min}(E, \players, U_{E, \players}, \tokens) > \indexshort_{C_{\epsilon}, \min}(E, \players, U'_{E, \players}, \tokens).
    \end{equation}
    
\end{thm}
\begin{proof}
Let $B$ be the largest voting bloc by token holdings before $T_{\texttt{bribe}}$. First, note that, after $T_{\texttt{bribe}}$, all voters in $\hat{\players}$ belong to the same voting bloc $B'$. Recall that, in our DAO abstraction, bribing a player $\player$ to flip its vote in election $e$ from $\false$ to $\true$ costs $\max(2 \cdot \util_{\player}(e,\false) + \epsilon,0)$. So, for every $P \in \hat{\players}$ and $e \in E$, either $\util_P(e, \true)$ was already positive to begin with, or it is now $|\util_{\player}(e,\false)| + \epsilon$. Then, since $\sgn{\util_\player(e, \true)} = 1$ for all $e \in E$, all of $\hat{P}$ consists of a single voting bloc $B'$ according to \clustershort. 

It follows then that $\tokens(B') \geq \tokens(B)$, as either the ``new'' voting bloc $B'$ is now the largest bloc, or the same bloc is the largest before and after $T_{\texttt{bribe}}$. Then, from Theorem~\ref{thm:vbe-master}, it follows that
\small
\[\indexshort_{C_{\epsilon}, \min}(E, \players, U_{E, \players}, \tokens) \geq \indexshort_{C_{\epsilon}, \min}(E, \players, U'_{E, \players}, \tokens)\]
\normalsize
as desired.
\end{proof}

\begin{thm}[Internal Bribery and \indexshort]
\label{thm:bribery2}
    Let $(E, U'_{E, \players}, \tokens) = T_{\texttt{bribe}}(\players, E, U_{E, \players}, \tokens)$ be the transformation where $U'_{E, \players}$ is some arbitrary change in the utilities of the players. The rest of the system remains unchanged. Assume that an entity in $\players$ needs to bribe other players holding a total of at least $n_1$ and $n_2$ tokens to guarantee an outcome of $\true$ in elections $E$ before and after $T_{\texttt{bribe}}$, respectively. Then, it follows that

    \begin{align*}
        n_1 > n_2 \iff &\indexshort_{C_{\epsilon}, \min}(E, \players, U'_{E, \players}, \tokens) \\ <\ &\indexshort_{C_{\epsilon}, \min}(E, \players, U_{E, \players}, \tokens).
    \end{align*}
    
\end{thm}

This result sheds light on the scale of bribery in the case where the briber is a malicious tokenholder a priori. Conversely, the briber may instead be some external entity. In this case, decentralization also raises the risk of systemic bribery: if there are large players in the system, the briber can directly coordinate with whales to achieve their desired election outcome. If, however, the DAO is highly decentralized, the outcome of the election depends on many stakeholders, which thus requires large-scale coordination among these. More formally:

\begin{thm}[External Bribery and \indexshort]
\label{thm:bribery3}
    Let $(E, U'_{E, \players}, \tokens) = T_{\texttt{bribe}}(\players, E, U_{E, \players}, \tokens)$ be the transformation where $U'_{E, \players}$ is some arbitrary change in the utilities of the players. The rest of the system remains unchanged. Let $n_1$ and $n_2$ be the minimum number of players that an external entity needs to corrupt to guarantee an outcome of $\true$ in elections $E$ before and after $T_{\texttt{bribe}}$, respectively. Then, it follows that

    \begin{align*}
        n_1 > n_2 \iff &\indexshort_{C_{\epsilon}, \min}(E, \players, U'_{E, \players}, \tokens) \\ <\ &\indexshort_{C_{\epsilon}, \min}(E, \players, U_{E, \players}, \tokens).
    \end{align*}
    
\end{thm}

\paragraph{Result \#7: quadratic voting.} Our formalism captures the relationship between quadratic voting and bribery (which has been informally identified by prior work~\cite{philquadvotingdaos}). We define ``small'' accounts to be, concretely, those whose fraction of the total tokens increases with quadratic voting in place, and thus have their impact amplified. More formally, we denote that a player $\player \in \players$ benefits from quadratic voting by $\qv(\player, \tokens) = 1$, where 

\[\qv(\player, \tokens) = 1 \iff \frac{\tokens(\player)}{\sum\limits_{p \in \players} \tokens(p)} < \frac{\sqrt{\tokens(\player)}}{\sum\limits_{p \in \players} \sqrt{\tokens(p)}}.\]

The relationship between quadratic voting and bribery hinges on whether the cost to bribe a player is the same with or without quadratic voting. 
Whether quadratic voting changes a player's utility or not will vary across systems. Broadly speaking, if DAO members take governance seriously and are invested in election outcomes, quadratic voting indeed changes utilities: since smaller accounts become more ``pivotal'' as a result of quadratic voting, their utilities increase correspondingly. Conversely, if members have little interest in governance, the fact that their vote can now have a greater impact in the election will not change their utilities. As such, the nature of a community must be taken into account when deciding to use quadratic voting.

\begin{thm}[Quadratic Voting and Bribery]
\label{thm:quadratic-voting}
Let $(E', U_{E', \players}, \tokens) = T_{\texttt{quad}}(\players, E, U_{E, \players}, \tokens)$ be the transformation where all elections $E$ employ quadratic voting. We denote the election corresponding to $e \in E$ by $e' \in E'$. Let $f$ and $f'$ be the fraction of total votes that a bribing entity is able to control for some fixed expenditure $t$ in elections $E$ and $E'$, respectively. Then, it follows that
\small
    \begin{equation*} 
        f < f' \iff \exists \hat{\players} \subseteq \players \mid \forall \player \in \hat{\players},\ \Big(\qv(\player, \tokens) = 1 \wedge U_{E, \player} = U_{E', \player}\Big)
    \end{equation*}
\end{thm}
\normalsize

This result thus shows that quadratic voting may be favorable for a bribing entity. In particular, if there are enough small voters whose utilities are unchanged, the cost to guarantee successful bribery decreases:
\begin{corollary}
Assume that, for $\hat{P}$ as defined in Theorem~\ref{thm:quadratic-voting}, $\tokens(\hat{P}) > q \cdot \sum_{\player \in \players} \tokens(P)$. Let $t$ and $t''$ be the expenditure required to guarantee an outcome of $\true$ in elections $E$ and $E'$, respectively. Then, it follows that $t' < t$.
\end{corollary}
This corollary simply follows from the fact that, as proved in Theorem~\ref{thm:quadratic-voting}, the expenditure $t'$ required to control a fraction of $q$ votes in $E'$, and thus guarantee successful bribery in $E'$, would only be enough to acquire a fraction of $q - \epsilon$ votes in $E$. As such, some additional expenditure is required to cross the threshold of $q$ votes.

\end{document}